\documentclass[natbib,authoryear]{svjour3}
\smartqed  
\usepackage{natbib}
\usepackage{color}

\usepackage{tikz}
\usepackage{cancel}

\usepackage{graphicx,subfigure}

\usepackage{rotating}
\usepackage{amssymb}
\usepackage{amsmath}
\usepackage{cancel}

\newtheorem{thm}{Theorem}

\newtheorem{rem}{Remark}

\newcommand{\figwidth}{0.7\textwidth}

\journalname{Journal of Mathematical Biology}

\begin{document}

\title{The Limiting Dynamics of a Bistable Molecular Switch With and Without Noise}

\titlerunning{The Limiting Dynamics of a Bistable Molecular Switch}

\author{Michael C. Mackey\and Marta Tyran-Kami\'nska}

\date{\today}

\institute{Michael C. Mackey \at Departments of Physiology, Physics \& Mathematics and Centre for Applied Mathematics in Bioscience and Medicine (CAMBAM), McGill University, 3655 Promenade Sir William Osler, Montreal, QC, CANADA, H3G 1Y6
\\ \email{michael.mackey@mcgill.ca}\and Marta Tyran-Kami\'nska \at Institute of Mathematics,
University of Silesia, Bankowa 14, 40-007 Katowice, POLAND\\ \email{mtyran@us.edu.pl} }


\maketitle

\begin{abstract}
We consider the dynamics of a population of organisms containing two mutually  inhibitory gene regulatory
networks, that can result in a bistable switch-like behaviour.  We completely characterize their local and
global dynamics in the absence of any noise, and then go on to consider the effects of either noise coming from
bursting (transcription or translation), or Gaussian noise in molecular degradation rates when there is a dominant slow variable in the
system.  We show analytically how the steady state distribution in the population can range from a single
unimodal distribution through a bimodal distribution and give the explicit analytic form for the invariant
stationary density which is globally asymptotically stable.   Rather remarkably, the behaviour of the stationary density with respect to the parameters characterizing the molecular behaviour of the bistable switch is qualitatively identical in the presence of noise coming from bursting as well as in the presence of Gaussian noise in the degradation rate.  This implies that one cannot distinguish between either the dominant source or nature of noise based on the stationary molecular distribution in a population of cells.  We finally show that the switch model with bursting but two dominant slow genes has an asymptotically stable stationary density.

\end{abstract}


\section{Introduction}\label{sec:intro}

In electrical circuits there are only two elementary ways to produce bistable behavior. Either with positive
feedback (e.g. A stimulates B and B stimulates A) or with double negative feedback (A inhibits B and B inhibits
A).  This elementary fact, known to all electrical engineering students, has, in recent years, come to the
attention of molecular biologists who have rushed to implicate one or the other mechanism as the source of
putative or real bistable behavior in a variety of biological systems.  (In a gene regulatory framework we might term the double positive feedback switch an inducible switch, while the double negative feedback switch could be called a repressible switch.)  Some laboratories have used
this insight to engineer {\it in vitro} systems to have bistable behavior and one of the first was \cite{Gardner2000}
who engineered repressible switch like behavior of the type we study in this paper.  Some especially well written  surveys are to be found in \cite{ferrell-2002}, \cite{tyson-2003}, and \cite{ferrell-2004}.

Gene regulatory networks are, however, noisy affairs for a variety of reasons and it is now thought that this noise may actually play a significant role in determining function \citep{eldar2010}.
In such noisy dynamical systems experimentalists
will often take a populational level approach and infer the existence of underlying bistable behavior based on
the existence of bimodal  densities of some molecular constituent over some range of experimental parameter
values.

From a modeling perspective there have been a number of studies attempting to understand the effects of noise on gene regulatory dynamics. The now classical \cite{kepler01} really laid much of the ground work for subsequent studies by its treatment of a variety of noise sources and their effect on dynamics. \cite{mty-2011} examined the effects of either bursting or Gaussian noise on both inducible and repressible operon models, and \cite{waldherr2010} looked at the role of Gaussian noise in an inducible switch model for ovarian follicular growth.

One of the most interesting situations is the observation that the presence of noise may induce bistability in a gene regulatory model when it was absolutely impossible to have bistable behaviour in the absence of noise.  This has been very nicely explored by \cite{artyomov07} (in competing positive/negative feedback motifs), and \cite{samoilov-05} (in enzymatic futile cycles),  while \cite{qian2009} and \cite{bishop2010}  analytically explored noise induced bistability, the latter in a phosphorylation-dephosphorylation cycle model.  \cite{vellela2008}  examined the role of noise in shaping the dynamics of the bistable Schl\"{o}gl chemical kinetic model.

For bistable repressible switch models \cite{wang-07} examined quorum-sensing with degradation rate noise in phage $\Lambda$ while \cite{morelli-08a}   examined the role of noise in protein production rates.  \cite{morelli-08b} carried out numerical studies of repressible switch slow dynamics in the face of noise.  \cite{bokes-13} gave a nice overview of the various approaches to the modeling of these systems and then  examined the role of transcriptional/translational bursting in repressible and inducible systems 
as well as in a repressible switch.
\cite{caravagna-13} examined the effects of bounded Gaussian noise on mRNA production rates in a repressible switch model, while \cite{strasser12} have looked at a model for the Pu/Gata switch (a repressible switch implicated in hematopoietic differentiation decision making) with high levels of protein and low levels of DNA.  

In this paper, we  extend the work of \cite{mty-2011} on inducible and repressible systems to an analytic consideration of an inducible switch in the presence of either bursting transcriptional (or translational) noise or Gaussian noise.
The paper is organized as follows. Section \ref{sec:generic} lays the groundwork by
developing the deterministic model  based on ordinary differential equations (a generalization of \cite{grigorov1967model}, the earliest study we know of, and  \cite{cherry-2000}) that we use to consider  the influence of noise.  This is followed in Section \ref{sec:ss-dynam}
with an analysis of the deterministic system, including the coexistence of multiple steady
states, and their stability.  This section, though superficially similar to the treatment of \cite{mty-2011}, extends their results to a completely different situation than previously considered, namely a model for a repressible switch. Section \ref{ssec:fast-slow} briefly considers how the
existence of fast and slow variables enables the simplification of the dynamics, and consequently  makes computations tractable, while the following Section
\ref{sec:dynamics-single} introduces bursting transcriptional or translational noise and derives the stationary
population density in a variety of situations when there is a single dominant slow variable.  We not only give explicit analytic expressions for these stationary densities, but also show that they are globally asymptotically stable. Section \ref{sec:dynamics-degrad} considers an alternative situation in which there is Gaussian distributed noise in the degradation rate for a single slow variable.
We again give the analytic form for the stationary densities as well as demonstrating their stability.  Section \ref{sec:2-dom-burst} expands on Section \ref{sec:dynamics-single} by considering bursting transcription or translation but in the situation where there are two dominant slow variables. The models in Sections \ref{sec:dynamics-single}-\ref{sec:2-dom-burst} are expressed as stochastic differential equations. The paper concludes
with a short discussion.

\section{The bistable genetic switch}\label{sec:generic}

\subsection{Biological background}\label{ssec:concept}

The paradigmatic molecular biology example of a bistable switch due to reciprocal negative feedback is the
bacteriophage (or phage) $\lambda$, which is a virus capable of
infecting {\it Escherichia coli} bacteria.  Originally described in \cite{jacob-1961} and very nicely treated in \cite{ptashne-1986}, it is but one of scores of  mutually inhibitory bistable switches that have been found since.

\subsection{Model development}\label{ssec:model-devel}

Figure \ref{fig:Fig1}  gives a cartoon representation of the situation we are  modeling here,  which is a generalization of the work of \cite{grigorov1967model} and \cite{cherry-2000}. { The original postulate for the hypothetical regulatory network of Figure \ref{fig:Fig1} is to be found in the lovely paper \citep{monod-jacob-61} which treats a number of different molecular control scenarios, and } the reader
may find reference to that figure helpful while following the model development below.  { It should be noted that with the advent of the power of synthetic biology it is now possible to construct molecular control circuits with virtually any desired configuration and thereby experimentally investigate their dynamics \citep{hasty2001}.}
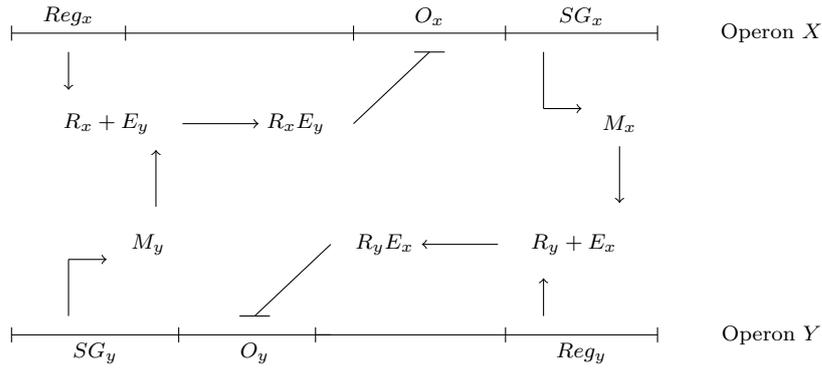
\begin{figure}[htb]
\begin{center}
\begin{tikzpicture}
    \draw (0,0) --(1.5,0) node[midway,above] {$Reg_x$};
    \draw (0,-0.1)--(0,0.1);
    \draw (1.5,-0.1)--(1.5,0.1);
    \draw (1.5,0) -- (4.5,0);
    \draw (4.5,-0.1)--(4.5,0.1);
    \draw (6.5,-0.1)--(6.5,0.1);
    \draw (4.5,0) --(6.5,0) node[midway,above] {$O_x$};
    \draw (6.5,0) --(8.5,0) node[midway,above] {$SG_x$};
    \draw (8.5,-0.1)--(8.5,0.1);
    \node at (10 ,0) {Operon $X$};
    \draw[->] (0.75,-0.25) -- (0.75,-.75);
    \node at (1.25,-1.2) {$R_x+E_y$};
    \draw[->] (2.25,-1.2)--(3.25,-1.2) node[right] {$R_xE_y$};
    \draw (4.5,-1.2) -- (5.5,-0.25);
    \draw (5.3,-0.25)-- (5.7,-0.25);
    \draw (0,-4) --(2.2,-4) node[midway,below] {$SG_y$};
    \draw (0,-4.1)--(0,-3.9);
    \draw (2.2,-4.1)--(2.2,-3.9);
    \draw (2.2,-4) -- (4.2,-4) node[midway,below] {$O_y$};
    \draw (4,-4.1)--(4,-3.9);
    \draw (4,-4)--(6.5,-4);
    \draw (6.5,-4.1) --(6.5,-3.9);
    \draw (6.5,-4) --(8.5,-4) node[midway,below] {$Reg_y$};
    \draw (8.5,-4.1) --(8.5,-3.9);
    \node at (10 ,-4) {Operon $Y$};
    \node at (8.0,-1.2) {$M_x$};
    \draw (7.0,-0.25) -- (7.0,-1.0);
    \draw[->] (7.0,-1.0)--(7.5,-1.0);
    \node at (7.4,-2.8) {$R_y+E_x$}; \draw[->] (7.0,-3.75)--(7.0,-3.25);
    \draw[->] (8.0,-1.5)--(8.0,-2.25); \draw[->] (6.4,-2.8)--(5.4,-2.8) node[left] {$R_yE_x$}; 
    \node at (1.8,-2.8) {$M_y$};
    \draw (4.2,-2.8) -- (3.2,-3.75);
    \draw (3,-3.75)-- (3.4,-3.75);
    \draw (0.75,-3.0) -- (0.75,-3.75);
    \draw[->] (0.75,-3.0)--(1.25,-3.0);
    \draw[->] (1.9,-2.3)--(1.9,-1.55);
\end{tikzpicture}
\end{center}
\caption{A schematic depiction of the elements of a bistable genetic switch, following \cite{monod-jacob-61}. There are two operons ($X$ and
$Y$).  For each, the regulatory region ($Reg_x$ or $Reg_y$) produces a repressor molecule ($R_x$ or $R_y$) that
is inactive unless it is combined with the effector produced by the opposing operon ($E_y$ or $E_x$
respectively).  In the combined form ($R_xE_y$ or $R_yE_x$) the repressor-effector complex binds to the operator
region ($O_x$ or $O_y$ respectively) and blocks transcription of the corresponding structural gene ($SG_x$ or
$SG_y$).  When the operator region is {\it not} complexed with the active form of the repressor, transcription
of the structural gene can take place and mRNA ($M_x$ or $M_y$) is produced.  Translation of the mRNA then
produces an effector molecule ($E_x$ or $E_y$). These effector molecules then are capable of interacting with
the repressor molecule of the opposing gene.   See \cite{monod-jacob-61}. } \label{fig:Fig1}
\end{figure}

\citet{polynikis09} offers a nice survey of techniques applicable to the approach we take in this section. We consider two operons $X$ and $Y$ such that the `effector' of $X$, denoted by $E_x$, inhibits the transcriptional production of mRNA from operon $Y$ and vice versa.  We take the approach of \cite{goodwin1965} as extended and developed in \citep{Griffith68a,Griffith68b,othmer76,selgrade79}. Consider initially a single operon $a$
where $a \in \{x,y\}$ and denote by $\bar a\in \{y,x\}$ the opposing operon.  For the mutually repressible systems we consider here, in the {\it presence} of the
effector molecule $E_{a}$   the repressor $R_{\bar a}$ is {\it active} (able to bind to
the operator region), and thus block DNA transcription.
The effector binds  with the inactive form $R_{\bar a}$ of the repressor, and when bound to the effector the repressor becomes active. We take
this reaction to be in equilibrium and of the form
\begin{equation}
R_{\bar a} + n_{\bar a}E_{a}  \rightleftharpoons R_{\bar a}E_{an_{\bar a}}  \label{e:equ}.
\end{equation}
Here, $R_{\bar a}E_{an_{\bar a}} $ is a repressor-effector complex and $n_{\bar a}$ is the number  of effector molecules that inactivate the repressor $R_{\bar a}$.
If we let the mRNA and  effector concentrations be denoted by $(M_a,E_a)$ then we assume
that the dynamics for operon $a$ are given by
   \begin{align}
    \dfrac{dM_a}{dt} &= \bar b_{d,a} \bar \varphi_{m,a}f_{a}( E_{\bar a}) -\gamma_{M_a} M_a, \quad \bar a \in \{y,x\} \label{eq:mrna}\\
    \dfrac{dE_a}{dt} &= \beta_{E_a}   M_a -\gamma_{E_a} E_a.\label{eq:effector}
    \end{align}
It is assumed in \eqref{eq:mrna} that the rate of mRNA production is proportional to the fraction of time the operator region is
active and that the maximum level
of transcription is $ \bar b_{d,a}$, and that the effector production rate is proportional to the amount of mRNA.  Note that  the production of $M_x$ is regulated by $E_y$ and vice versa, and that the components $(M_a,E_a)$ are subject to degradation\footnote{The more precise form for \eqref{eq:effector} would be
$$
\dfrac{dE_a}{dt} = \beta_{E_a}   M_a -\gamma_{E_a} E_a
    -n_{\bar a} k_{1,\bar a}R_{\bar a}\cdot E_a^{n_{\bar a}} + n_{\bar a} k_{-1,\bar a} [R_{\bar a}E_{an_{\bar a}}]
$$
where $k_{1,\bar a}/k_{-1,\bar a}$ is the equilibrium constant.  The equilibrium assumption means that the last two terms cancel.}. The function $f$ is
calculated next.

To compute $f$ we temporarily suppress the subscript $a$ and then restore it at the end. Let the corresponding reaction in \eqref{e:equ} and the equilibrium constant be
\[
R + nE \stackrel{K_1} \rightleftharpoons RE_n \qquad K_1 = \frac{RE_n}{R \cdot E^n}.
\]
 There is an interaction between the  operator $O$ and
repressor $R$ described by
\begin{equation*}
O + RE_n \stackrel{K_2}\rightleftharpoons ORE_n \qquad K_2 = \frac{ORE_n} {O
\cdot RE_n}. \label{or-rep}
\end{equation*}
The  total operator is given by
\begin{equation*}
O_{tot} = O + ORE_n = O +   K_1 K_2 O \cdot R\cdot E^n   = O(1+  K_1 K_2R\cdot E^n), \label{rep-totoper}
\end{equation*}
while the total repressor $R_{tot}$ is
\begin{equation*}
R_{tot} = R + K_1 R \cdot E^n + K_2 O \cdot RE_n, \label{ind-totrepre}
\end{equation*}
so the fraction of operators not bound by repressor is given by
\begin{equation*}
f(E) = \frac{O}{O_{tot}} = \frac{1}{1+   K_1 K_2R\cdot E^n   }. \label{rep-frac1}
\end{equation*}
 If the amount  repressor bound to the operator is small compared to the total amount of repressor then  { $R_{tot} \simeq R(1 + K_1  \cdot E^n )$ and consequently}
 \begin{equation*}
f(E) = \frac{1+ K_1E^n}{1+(K_1+   K_1 K_2R_{tot}   )E^n  } = \frac{1+
K_1E^n}{1 + KE^n}, \label{rep-frac2}
\end{equation*}
where  $K =    K_1(1+ K_2R_{tot}) $. When $E$ is large there will be maximal repression, but even then there will still
be a basal level of mRNA production proportional to $K_1 K^{-1}<1$ (this is known as leakage). The variation of
the DNA transcription rate with effector level is given by $\varphi = \bar \varphi_{m} f$ or
    \begin{equation}
    \varphi(E) = \bar \varphi_{m}  \frac{1+
    K_1E^n}{1 + KE^n}  =
    \bar \varphi_{m}f(E),
    \label{rep-phi}
    \end{equation}
where  $\bar \varphi_{m}$ is the maximal DNA transcription rate (in
units of inverse time).

Now explicitly including the proper subscripts we have
    \begin{equation*}
    \varphi_a(E_{\bar a}) = \bar \varphi_{m,a}  \frac{1+
    K_{1,a}E_{\bar a}^{n_a}}{1 + K_aE_{\bar a}^{n_a}}  =
    \bar \varphi_{m,a}f_a(E_{\bar a}),
    \label{rep-phi-subscript}
    \end{equation*}
where { $K_a = K_{1,a}(1 + K_{2,a} R_{tot,a})$}.

We next rewrite Equations
\ref{eq:mrna}-\ref{eq:effector}   by defining dimensionless
concentrations.  Equation
\ref{rep-phi} becomes
    \begin{equation*}
    \varphi_a(e_{\bar a}) =
    \varphi_{m,a}f_a(e_{\bar a}),
    \label{eq:gen-response-fun-dimen}
    \end{equation*}
where the dimensionless rate $\varphi_{m,a}$ is defined by
    \begin{equation*}
    \varphi_{m,a} = \dfrac {\bar \varphi_{m,a} \beta_{E,a} }{\gamma_{M,a} \gamma_{E,a} } \quad \mbox {and} \quad f_a(e_{\bar a})=
    \dfrac{1+e_{\bar a}^{n_a}}{1 + \Delta_a e_{\bar a}^{n_a}},
    \label{eq:gen-response-fun-dimen-f}
    \end{equation*}
$\Delta_a = K_a K_{1,a}^{-1}$, and the
dimensionless effector concentration $(e_a)$ is defined by
    \begin{equation*}
    E_a = \eta_a e_{\bar a}   \quad \mbox{with} \quad\eta_a = \dfrac {1}{\sqrt[n_a]{K_{1,a}}}.
    \end{equation*}
Recall that $\Delta_a^{-1}$ denotes the leakage and note that if $\Delta_a$ goes to infinity then the transcription goes to zero. 
Similarly using a dimensionless  mRNA
concentration ($m_a$) given by
    \begin{equation*}
      M_a = m_a
    \eta_a \dfrac{\gamma_{E_a}  }{\beta_{E_a}  },
    \end{equation*}
Equations \ref{eq:mrna}-\ref{eq:effector} take the form
 \begin{align*}
    \dfrac{dm_a}{dt} &= \gamma_{M_a} [ \kappa_{d,a} f_a( e_{\bar a}) - m_a],\\
    \dfrac{de_a}{dt} &= \gamma_{E_a} (m_a - e_a),
    \end{align*}
with
$$
    \kappa_{d,a} = b_{d,a} \varphi_{m,a}, \quad \mbox{and} \quad b_{d,a} =\dfrac {\bar b_{d,a}}{\eta_a}
$$
which are both dimensionless.

Thus the equations governing the dynamics of this system are given by the four differential equations
\begin{align*}
    \dfrac{dm_x}{dt} &= \gamma_{M_x} [ \kappa_{d,x} f_x( e_{y}) - m_x],\\
      \dfrac{de_x}{dt} &= \gamma_{E_x} (m_x - e_x),\\
    \dfrac{dm_y}{dt} &= \gamma_{M_y} [ \kappa_{d,y} f_y( e_x) - m_y],\\
     \dfrac{de_y}{dt} &= \gamma_{E_y} (m_y - e_y)
    \end{align*}
 where
 $$
  f_x(e_{y})=
    \dfrac{1+e_{y}^{n_x}}{1 + \Delta_x e_{y}^{n_x}} \quad \mbox{and} \quad
 f_y(e_{x})=
    \dfrac{1+e_{x}^{n_y}}{1 + \Delta_y e_{x}^{n_y}}.
$$

To make the model equations somewhat more straightforward, denote dimensionless
concentrations by $(m_x, e_x,m_y, e_y) = (x_1,x_2, y_1,y_2 )$ (with obvious changes in the other subscripts) to obtain
    \begin{align}
    \dfrac{dx_1}{dt} &= \gamma_{x_1} [\kappa_{d,x} f_x( y_2) -x_1], \label{eq:xmrna1-dimen}\\
    \dfrac{dx_2}{dt} &= \gamma_{x_2}  (  x_1 - x_2), \label{eq:xintermed2-dimen}\\
    \dfrac{dy_1}{dt} &= \gamma_{y_1} [\kappa_{d,y} f_y( x_2) -y_1], \label{eq:ymrna1-dimen}\\
    \dfrac{dy_2}{dt} &= \gamma_{y_2}  (  y_1 - y_2), \label{eq:yintermed2-dimen}
       \end{align}
Throughout, $\gamma_{\cdot}  $  is a  decay
rate  (time$^{-1}$), and so Equations
\ref{eq:xmrna1-dimen}-\ref{eq:yintermed2-dimen} are not dimensionless.  In addition to the loss rates explicitly appearing, we have the parameters $\kappa_{d,x}, \kappa_{d,y}$.  Since
\begin{equation}
 f_x(y_2)=
    \dfrac{1+y_2^{n_x}}{1 + \Delta_x y_2^{n_x}} \quad \mbox{and} \quad
 f_y(x_2)=
    \dfrac{1+x_2^{n_y}}{1 + \Delta_y x_2^{n_y}},
 \label{eq:control-dimen}
 \end{equation}
we have as well the four parameters $\Delta_x, \Delta_y, n_x, n_y$ to consider.  Note that
\begin{equation*}
f_x(0) = 1, \, \lim_{y_2 \to \infty} f_x(y_2) = \Delta_x^{-1} <1, \,f_y(0) = 1, \,   \lim_{x_2 \to \infty} f_y(x_2) = \Delta_y^{-1} <1.
\end{equation*}

\section{Steady states and dynamics}\label{sec:ss-dynam}
The dynamics of this model for a bistable switch can be analyzed as follows.  This section is an elaboration of aspects of the work presented in \cite{cherry-2000}. Set $W=(x_1,x_2,y_1,y_2)$ so the
system (\ref{eq:xmrna1-dimen})-(\ref{eq:yintermed2-dimen}) generates a flow $S_t(W)$.   The flow $S_t(W^0) \in
\mathbb{R}_4^+$ for all initial conditions $W^0=(x_1^0,x_2^0,,y_1^0,y_2^0) \in \mathbb{R}_4^+$  and $t>0$.

The steady states of the system (\ref{eq:xmrna1-dimen})-(\ref{eq:yintermed2-dimen})  are given by
$x_1^*=x_2^*=x^*,y_1^*=y_2^*=y^*$ where $(x^*,y^*)$ is the solution of
    \begin{align}
    x_1 &= x_2  = \kappa_{d,x}f_x(y_2) \label{eq:ss1} \\
    y_1 &= y_2 = \kappa_{d,y}f_y(x_2). \label{eq:ss2}
     \end{align}
For each solution $(x^*,y^*)$ of (\ref{eq:ss1})-(\ref{eq:ss2}) there is a  steady state $W^*$ of the model, and
the parameters $(\kappa_{d,x}, \kappa_{d,y},\Delta_x, \Delta_y, n_x, n_y)$ will determine whether $W^*$ is
unique or has multiple values.

\subsection{Graphical investigation of the steady states}
Figure \ref{fig:fig2} gives a graphical picture of the five qualitative possibilities for steady state solutions
of the pair of equations (\ref{eq:xmrna1-dimen})-(\ref{eq:yintermed2-dimen}).

\begin{figure}[htb]\centering
\includegraphics[width=\figwidth]{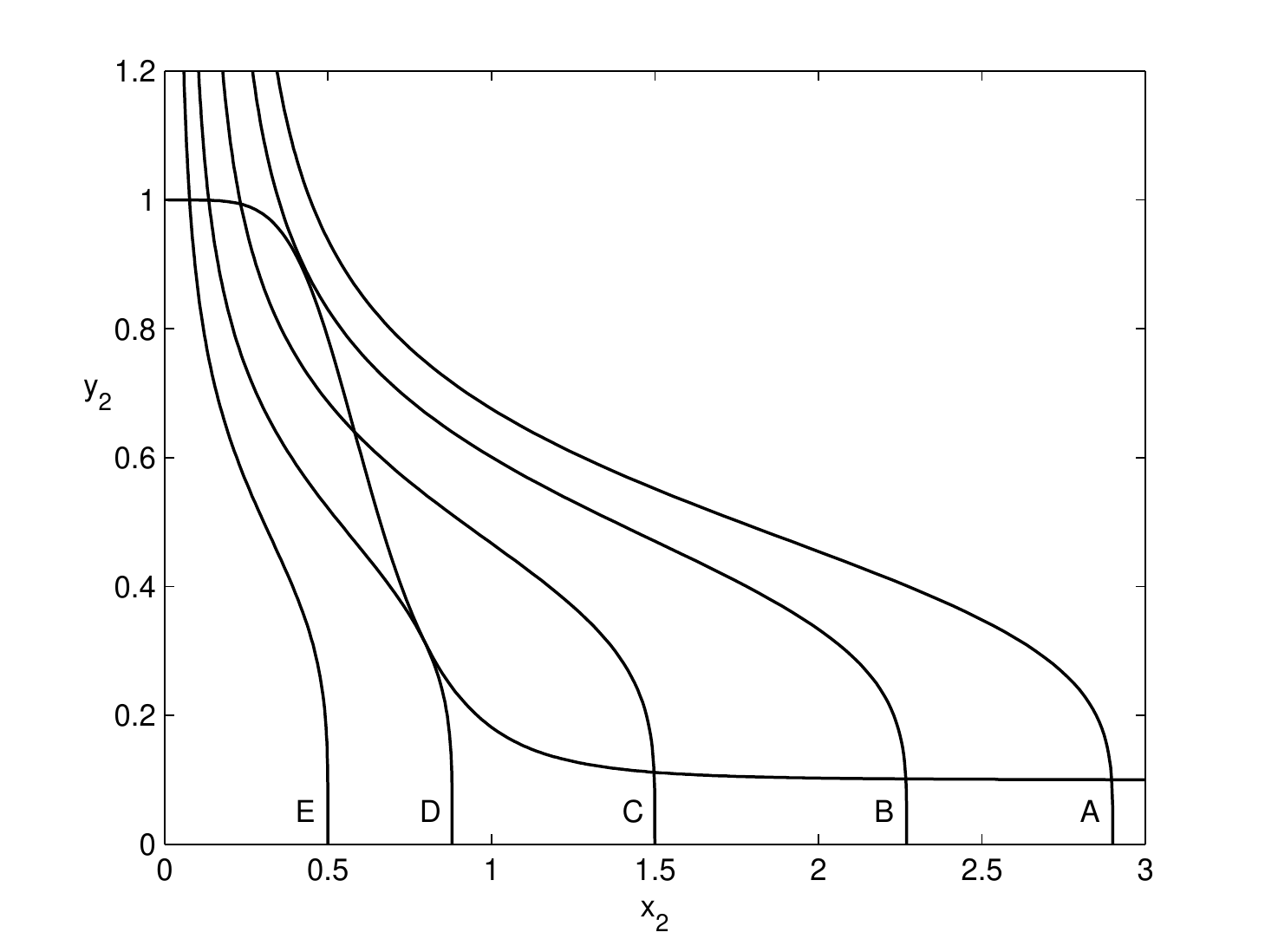}
\caption{A graphical representation of the possible steady state solutions of Equations \ref{eq:ss1} and
\ref{eq:ss2}.  We have plotted the $y_1$ and $x_1$ isoclines ($ y_2 = \kappa_{d,y}f_y(x_2)$ and $x_2 = \kappa_{d,x}f_x(y_2)$ respectively), and
assumed that the $y_1$ isocline (the graph of $ y_2 = \kappa_{d,y}f_y(x_2)$)  is not changed but that $x_1$ isocline (the graph of $x_2 =
\kappa_{d,x}f_x(y_2)$) is varied as indicated by the labels A to E, {\it e.g.} by decreasing  $\kappa_{d,x}$. (A) There is
a single steady state at a large value of $x_2$ and a correspondingly small value of $y_2$.  In this case operon
$X$ of the bistable switch is in the ``ON" state while operon $Y$ is in the ``OFF" state. This steady state is
globally stable.  (B) A decrease in  $\kappa_{d,x}$ now leads to a situation in which there are two steady
states, the largest (locally stable one) corresponding to the intersection of the two graphs, and the second
smaller (half stable) one where the two graphs are tangent.  (C) Further decreases in $\kappa_{d,x}$ now result
in three steady states.  For the largest (locally stable) one the operon $X$ is in the on state while $Y$ is in
the off state.  The smallest one (also locally stable) corresponds to operon $Y$ in the ON state and $X$ is in
the OFF state.  The intermediate steady state is unstable.  (D) This case is like B in that there are two steady
states, one (locally stable) defined by the intersection of the two graphs in which $Y$ is ON and the second at
the tangency of the two graphs is again half stable.   (E)  Finally, for sufficiently small $\kappa_{d,x}$ there
is a single globally stable steady state in which $Y$ is ON and $X$ is OFF.}\label{fig:fig2}
\end{figure}

An alternative, but equivalent, way of examining the steady state of this model is by examining the solution of either one of the pair of equations
\begin{equation*}
\dfrac{x}{\kappa_{d,x}} = f_x(\kappa_{d,y}f_y(x)) := {\mathcal F}_x (x), \,\,
\dfrac{y}{\kappa_{d,y}} = f_y(\kappa_{d,x}f_x(y)) := {\mathcal F}_y(y).
\end{equation*}
We choose to deal with the first.  Note that since both $f_x$ and $f_y$ are monotone decreasing functions of their arguments, the composition of the two
\begin{equation}
{\mathcal F}_x(x)   = \dfrac{1+(\kappa_{d,y}f_y(x))^{n_x}}{1 + \Delta_x (\kappa_{d,y}f_y(x))^{n_x}}
= \dfrac{1+\left (\kappa_{d,y} \dfrac{1+x^{n_y}}{1 + \Delta_y x^{n_y}} \right )^{n_x}}{1 + \Delta_x \left (\kappa_{d,y} \dfrac{1+x^{n_y}}{1 + \Delta_y x^{n_y}} \right )^{n_x}}
\label{eq:composition-x}
\end{equation}
is a monotone increasing function of $x$ with
$$
{\mathcal F}_x(0)   =
\dfrac{1+ \kappa_{d,y}    ^{n_x}}{1 + \Delta_x  \kappa_{d,y}    ^{n_x}} := {\mathcal F}_{x,0}
$$
and
$$
\lim_{x \to \infty} {\mathcal F}_x(x)   =
\dfrac{1+ (\kappa_{d,y}  \Delta_y^{-1}    )^{n_x}}{1 + \Delta_x   (\kappa_{d,y}   \Delta_y^{-1}    )^{n_x}} := {\mathcal F}_{x,\infty}
> {\mathcal F}_{x,0}.
$$
In Figure \ref{fig:Fig3} we have shown graphically the same sequence of steady states as we illustrated in Figure \ref{fig:fig2}
\begin{figure}[tb]
\centering
\includegraphics[width=\figwidth]{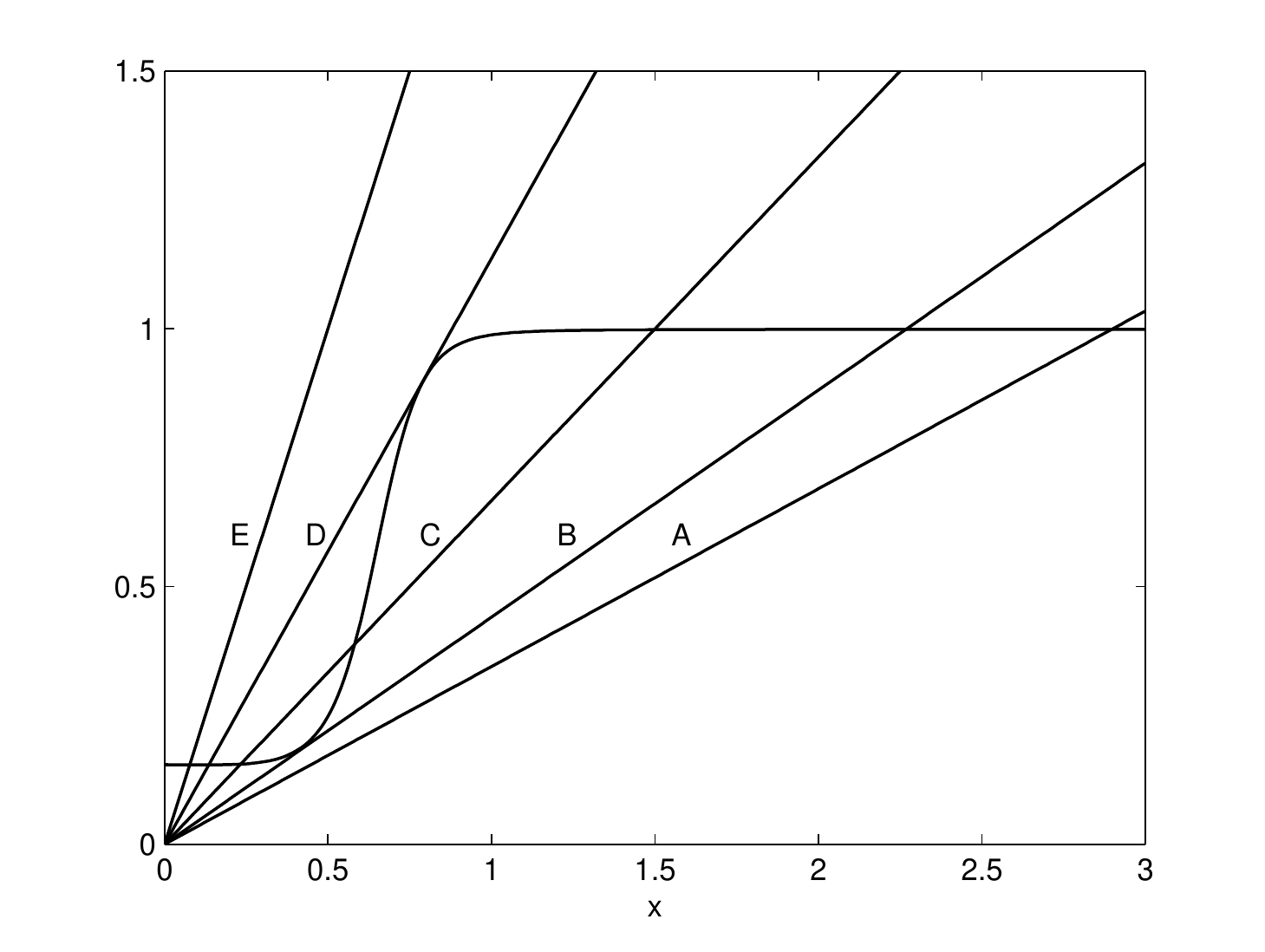}

\caption{A graphical representation of the possible steady state solutions of the equation ${x}/{\kappa_{d,x}} = f_x(\kappa_{d,y}f_y(x)) := {\mathcal F}_x (x)$.  The smooth monotone increasing graph is that of ${\mathcal F}_x (x)$ as given in Equation  \ref{eq:composition-x}, while the straight line is that of ${x}/{\kappa_{d,x}}$ for different values of $\kappa_{d,x}$.  The five straight lines correspond to the five possibilities (A through E) in Figure \ref{fig:fig2}.
} \label{fig:Fig3}
\end{figure}

\subsection{Analytic investigation of the steady states}

\noindent{\bf Single versus multiple steady states.} This model for a bistable genetic switch may
have one [$W_1^*$ ( E of Figure \ref{fig:fig2} or  Figure \ref{fig:Fig3}) or $W_3^*$ (A)], two [$W_1^*,W_2^*=W_3^*$ (D) or
$W_1^*=W_2^*,W_3^*$ (B)], or three [$W_1^*,W_2^*,W_3^*$ (C)] steady
states, with the ordering $0 \preceq  W_1^* \preceq W_2^* \preceq W_3^*$,
indicating that $W_1^*$ corresponds to operon $X$ in the OFF state and operon $Y$ in the ON state while at $W_3^*$ $X$ is ON and $Y$ is OFF.

Analytic conditions for the existence of one or more steady states can be obtained by first noting that we must
have \begin{equation} 
\dfrac{x}{\kappa_{d,x}} = f_x(\kappa_{d,y}f_y(x)) := {\mathcal F}_x (x) \label{eq:ss-equality} 
\end{equation}
satisfied. In Figure \ref{fig:Fig4} we have illustrated Equation \ref{eq:ss-equality} for various values of
parameters.
\begin{figure}[tb]
\centering
\includegraphics[width=\figwidth]{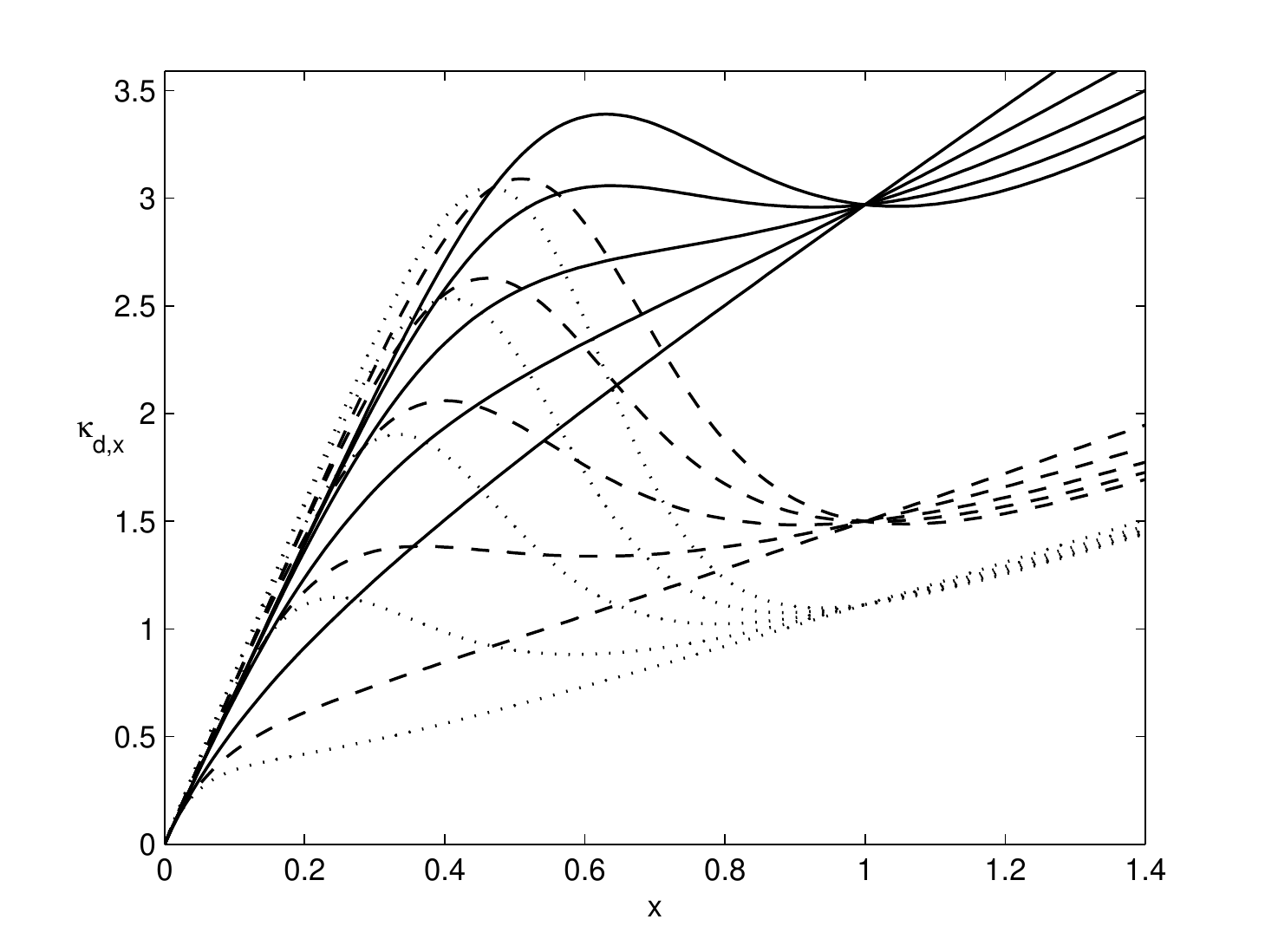}

\caption{The plot of $\kappa_{d,x}$ versus $x$ obtained from Equation  \ref{eq:ss-equality}. The figure was
constructed for the following parameters: $n_x\in \{1,2,3\}$, $n_y\in \{1,2,3,4,5\}$,  $\kappa_y=2$,
$\Delta_x=12$, $\Delta_y=10$. The solid lines correspond to $n_x=1$ and we increase $n_y$ from $1$ (the lowest
line) to $5$ (the top one). The dashed lines correspond to $n_x=2$ and the dotted to $n_x=3$.} \label{fig:Fig4}
\end{figure}

In addition to this criteria, we have a second relation at our disposal at the delineation points
between the existence of two and three steady state.  These points are also determined by a second relation
since $x/\kappa_d$ is tangent to ${\mathcal F}_x (x)$ (see Figure \ref{fig:Fig3} B,D). Thus we must also have
    \begin{equation*}
    \dfrac 1 {\kappa_{d,x} }
    = \dfrac{d {\mathcal F}_x (x) }{dx }.
    \label{eq:ss-slope}
    \end{equation*}
Now the problem is to derive values for $x_\pm$ at which a tangency occurs, as well as to figure out some way to make a parametric plot of a combination of $\kappa_{d,x}, \kappa_{d,y}, \Delta_x, \Delta_y$ for given values of $n_x,n_y$.

Indeed, from Equations \ref{eq:ss1} and \ref{eq:ss2} we have
\begin{equation}
x = \kappa_{d,x}f_x(y) \quad \mbox{and} \quad y = \kappa_{d,y}f_y(x).
\label{eq:kappa(x)}
\end{equation}
Additionally at a tangency between $f_x(y)$ and $f_y(x)$ we must have
$$
x = \kappa_{d,x}f_x(\kappa_{d,y}f_y(x)),
$$
so
$$
1 = \kappa_{d,x} \kappa_{d,y} f'_x f'_y.
$$
However,
$$
\kappa_{d,x} \kappa_{d,y} = \dfrac{xy}{f_x(y) f_y(x)}
$$
so we have an implicit relationship between $x$ and $y$ given by
$$
\dfrac{f_x(y)}{y f'_x(y)} = \dfrac{xf'_y(x)}{f_y(x)}
$$
that, when written explicitly becomes
\begin{equation}
L(y) := - \dfrac{(1+y^{n_x})(1+\Delta_xy^{n_x})}{n_x(\Delta_x -1)y^{n_x}}
=
- \dfrac
{n_y(\Delta_y -1)x^{n_y}}
{(1+x^{n_y})
(1+\Delta_y x^{n_y} ) } := R(x).
\label{eq:L=R}
\end{equation}

Now $L(y)$ has a maximum  at $y_{max} =  \Delta_x^{-1/2n_x}$ and
$$
L(y_{max}) := L_{max} = - \dfrac{(1+\sqrt{\Delta_x})^2}
{n_x(\Delta_x -1)},
$$
while $R(x)$ has a minimum at $x_{min} = \Delta_y^{-1/2n_y}$ given by
$$
R(x_{min}) := R_{min} = - \dfrac{n_y(\Delta_y -1)}{(1+\sqrt{\Delta_y})^2}.
$$
A necessary condition for there to be a solution to Equation \ref{eq:L=R}, and thus a necessary condition for
bistability,  is that $L_{max} \geq R_{min}$ or
\begin{equation*} n_x n_y \geq
\dfrac{(1+\sqrt{\Delta_x})^2(1+\sqrt{\Delta_y})^2}{(\Delta_x -1)(\Delta_y -1)} \geq 1.
\end{equation*} This is
interesting in the sense that if either $n_x$ OR $n_y$ is one but the other is larger than one then the
possibility of bistability behavior still persists, while in the situation of \cite{mty-2011} this is
impossible   (the same observation has been made by \cite{cherry-2000} in a somewhat simpler model).  However, note from Figure \ref{fig:Fig4} that this necessary condition is far from what is sufficient since it would appear from Equation \ref{eq:ss-equality} that a necessary and sufficient condition is more like $n_xn_y \simeq 4$.

Going back to Equation \ref{eq:L=R}, we can write
\begin{equation*}
 \Delta_x y^{2n_x} + [n_x(\Delta_x
-1)R(x) + (\Delta_x + 1)] y^{n_x} + 1 = 0, \label{eq:y(x)}
\end{equation*}
which has two positive solutions $y_{\pm}$ given by
\begin{equation}
y_{\pm} = \sqrt[n_x]{ \dfrac{\Delta_x - 1}{2 \Delta_x}
    \left\{    -
    n_x R(x)  -\dfrac { \Delta_x + 1}{ \Delta_x - 1}
        \pm \sqrt
    {
    [n_x R(x)]^2
    + 2n_x R(x) \dfrac { \Delta_x + 1 }{  \Delta_x - 1}     +1
     }
    \right\}},
    \label{eq:tangency}
\end{equation} provided that
\[
 [n_x R(x)]^2
    + 2n_x R(x) \dfrac { \Delta_x + 1 }{  \Delta_x - 1}     +1\ge 0
\]
and
\[
-n_x R(x)  -\dfrac { \Delta_x + 1}{ \Delta_x - 1}
        -\sqrt
    {
    [n_x R(x)]^2
    + 2n_x R(x) \dfrac { \Delta_x + 1 }{  \Delta_x - 1}     +1}\ge 0.
\]
Substitution of the result into Equations \ref{eq:kappa(x)} gives explicitly
\begin{equation}
\kappa_{d,x}(x) = \dfrac {x}{f_x(y(x))} \quad \mbox{and} \quad \kappa_{d,y}(x) = \dfrac{y(x)}{f_y(x)},
\label{eq:kappa-explicit}
\end{equation}
where $y(x)$ is either $y_{+}$ or $y_{-}$ as given by
\eqref{eq:tangency}.

\begin{figure}[tb]
\centering
\includegraphics[width=\figwidth]{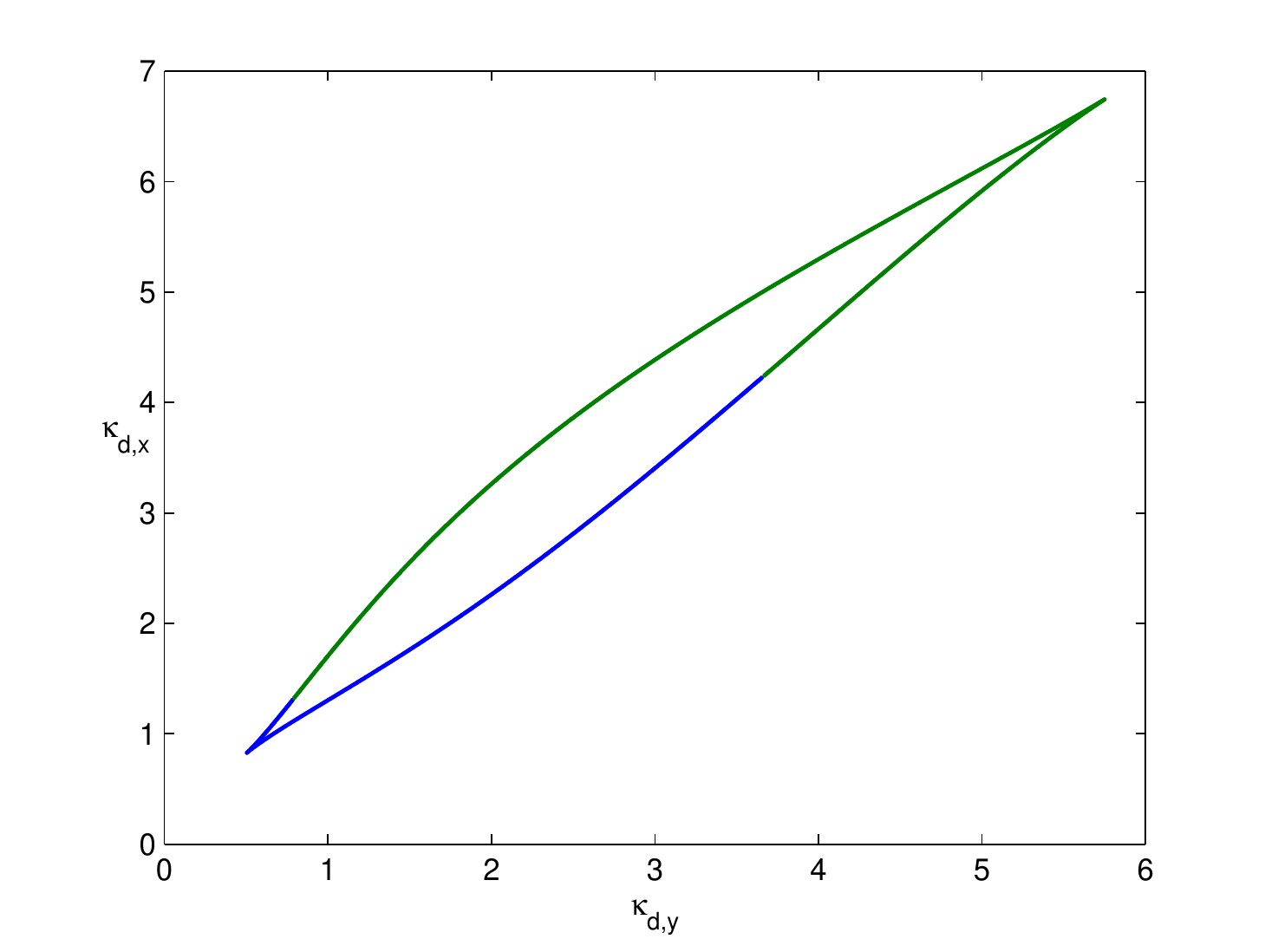}

\caption{The parametric plot of $\kappa_{d,x}$ versus $\kappa_{d,y}$ obtained from Equation
\ref{eq:kappa-explicit} where we used the following parameters: $n_x=2$, $n_y=3$,
 $\Delta_x=12$, $\Delta_y=10$. The blue line is for $y(x)=y_{-}$ and the green for $y(x)=y_{+}$. }
\label{fig:Fig5}
\end{figure}
In Figure \ref{fig:Fig5} we  have  plotted $\kappa_{d,x}(x)$ versus $\kappa_{d,y}(x)$ with
$x$ as the parametric variable.  Inside the region bounded by the blue line (below) and green line (above) we are assured of the existence of bistable behaviour while outside this region there will be only a single globally stable steady state.  Thus, for example, for a constant value of $\kappa_{d,y}$ such that bistability is possible, then increasing $\kappa_{d,x}$ from $0$ there will be a minimal value $\kappa_{d,x-}$ at which bistability is first seen and this will persist as $\kappa_{d,x}$ is increased until a second value $ \kappa_{d,x-} < \kappa_{d,x+}$ is reached where the bistable behaviour once again disappears. In Figure \ref{fig:Fig5-2} we have shown how the change of the parameter $\Delta_y$ influences the shape and position of the region of  parameters $\kappa_{d,y}$ and $\kappa_{d,x}$ where a bistable behaviour is possible. It is clear that an increase in $\Delta_y$ corresponds to a decrease in the leakage, and our results show a clear expansion in the size of the region of bistability as well as a shift in $(\kappa_{d,y}, \kappa_{d,x})$ space.  This is the same observation made in \cite{mty-2011}.
\begin{figure}[tb]
\centering
\includegraphics[width=\figwidth]{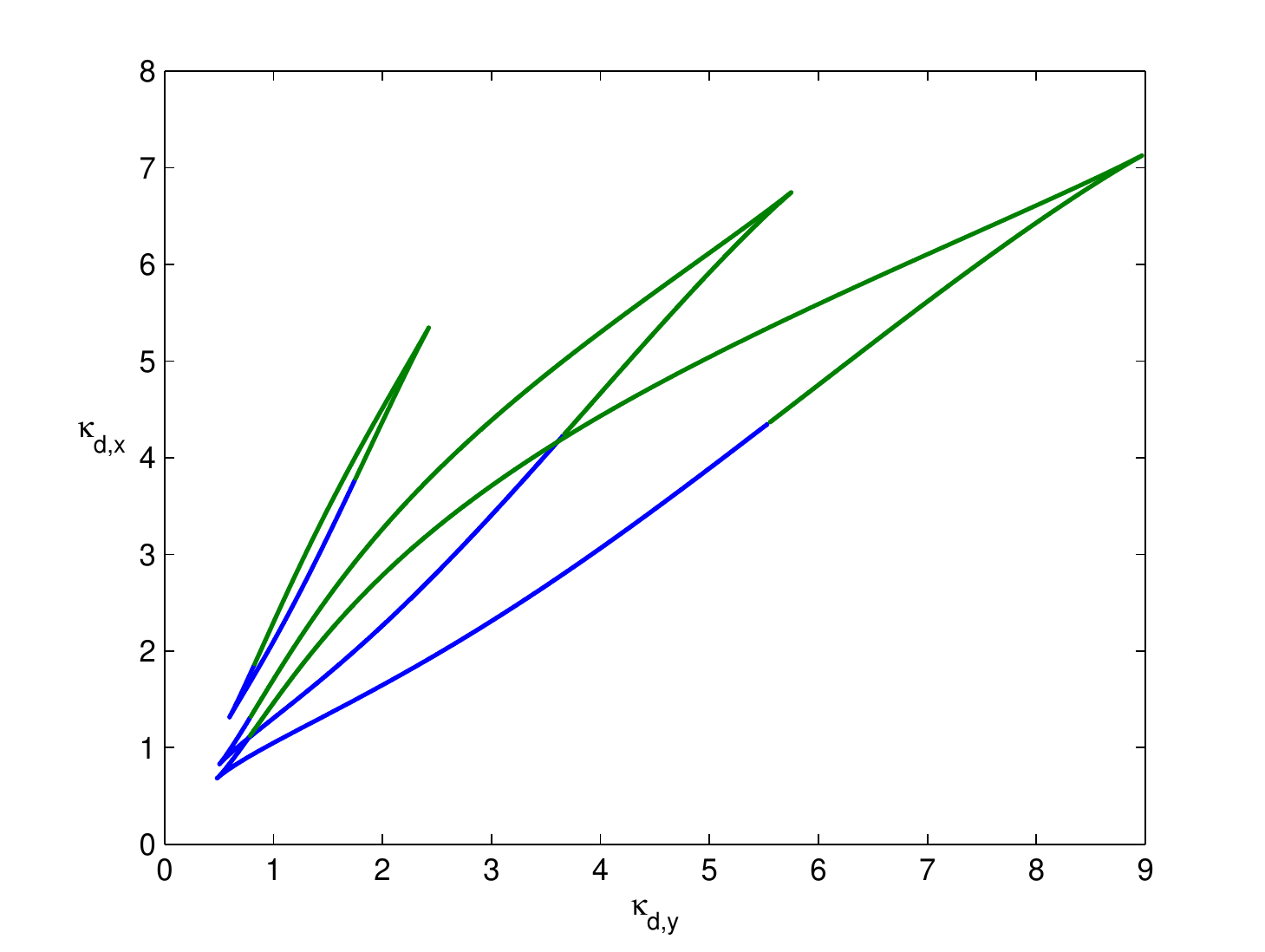}

\caption{ As in Figure~\ref{fig:Fig5} but with varying parameter $\Delta_y\in \{5,10,15\}$, from left to right.  }
\label{fig:Fig5-2}
\end{figure}

\subsection{Local and global stability.}

Whether or not a steady state $W^*$ is locally stable is completely determined by the eigenvalues that solve the equation
    \begin{equation}
    \prod_{i=1}^2 (\lambda + \gamma_{x_i})(\lambda + \gamma_{y_i}) -
    \prod _{i=1}^2 \gamma_{x_i} \gamma_{y_i} \kappa_{d,x} \kappa_{d,y}  f'_{x*} f'_{y*} =0,
    \label{eq:eigenvalue}
    \end{equation}
where $f'_{x*} = f_x'(x^*), f'_{y*} = f_y'(y^*)$.
Equation \ref{eq:eigenvalue} can be rewritten in the form
    \begin{equation}
    \sum _{i=1}^4 a_i \lambda^i  + a_0 = 0
    \label{eq:eigenvalue1}
    \end{equation}
where the $a_i$, $i > 0$ are positive and
$
a_0 = (1- \kappa_{d,x} \kappa_{d,y}  f'_{x*} f'_{y*}) \prod _{i=1}^2 \gamma_{x_i} \gamma_{y_i}.
$
By Descartes's rule of signs, (\ref{eq:eigenvalue1}) has no positive roots for $f'_{x*} f'_{y*} \in [0, (\kappa_{d,x} \kappa_{d,y})  ^{-1})$ or one
positive root otherwise. Denote a locally stable steady state by S, a half or
neutrally stable steady state by HS, and unstable steady state by US.
Then we know that there will be:
\begin{itemize}
    \item A single steady state $W_1^*$ (S), for $ \kappa_{d,x}
        \in [0, \kappa_{d,x-})$
    \item Two steady states $W_1^*$ (S) and
        $W_2^*=W_3^*$ (HS) for $\kappa_{d,x} = \kappa_{d,x-}$
    \item Three steady states $W_1^*$ (S), $W_2^*$
        (US), $W_3^*$ (S) for $\kappa_{d,x} \in
        (\kappa_{d,x-},\kappa_{d,x+})$
    \item Two steady states $W_1^*=W_2^*$ (HS) and $W_3^*$ (S)  for
        $\kappa_{d,x} = \kappa_{d,x+}$
    \item One steady state $W_3^*$ (S) for $\kappa_{d,x+} <
        \kappa_{d,x}$.
\end{itemize}

Global stability results of others complement this classification.
\begin{thm}\citep[Proposition 2.1, Chapter 4]{othmer76,smith95}
For the bistable switch given by Equations
\ref{eq:xmrna1-dimen}-\ref{eq:control-dimen}, define $I_x = [\kappa_{d,x}\Delta_x^{-1},1]$ and $I_y = [\kappa_{d,y}\Delta_y^{-1},1]$.  There is an attracting box
$B  \subset \mathbb{R}_4^+$, where
    \begin{equation*}
    B = \{(x_i,y_i): x_{1,2} \in I_x, y_{1,2} \in I_y   \},
    \end{equation*}
for which the flow $S_t$ is directed inward
 on the surface of $B$. All $W^*
    \in B$ and
    \begin{enumerate}
    \item If there is a single steady state,
        then it is globally stable.
    \item If there are two locally stable steady states,
        then all flows $S_t(W^0)$
        are attracted to one of them.
    \end{enumerate}
\end{thm}

\section{Fast and slow variables}\label{ssec:fast-slow}

Identification of fast and slow
variables in systems can often be used to achieve simplifications that allow quantitative examination of the relevant dynamics, and particularly to examine the approach to a steady state and the nature of that steady state.
A fast variable is one that relaxes much more rapidly to
an equilibrium than does a slow variable \citep{haken83}.  In chemical systems this separation is often
a consequence of differences in degradation rates, and the fastest variable  is the one with the largest degradation
rate. { In recent years, with the advent of synthetic biology, investigators have engineered a variety of gene regulatory circuits, including bistable switches of the type considered here, see \cite{hasty2001,huang2012}, in which they were able to experimentally control the speed with which particular variables approached a quasi-equilibrium state.  Thus this experimental technique offers an experimental way to actually achieve the simplification of causing particular variables to become fast variables.
We will use this technique analytically  in examining the effects of noise,  which has the added advantage of allowing us to derive analytic insights from the simplified model that seem to be impossible in the full model.}

If it is the case that there is a single dominant slow variable  in the system (\ref{eq:xmrna1-dimen})-(\ref{eq:yintermed2-dimen}) relative to all of the other three (and here we assume without loss of generality that it is in the $X$ gene) then the four variable system describing the full switch reduces to a single equation
\begin{equation}
\dfrac{dx}{dt} = \gamma[\kappa_{d,x} {\mathcal F} (x) -x],
\label{eq:one-slow}
\end{equation}
and $\gamma$ is the dominant (smallest) degradation rate.  (Here, and subsequently, to simplify  the notation we will drop the subscript $x$ whenever there will not be any confusion  when  treating the situation with a single dominant slow variable.)

\section{Transcriptional and translational bursting}\label{sec:dynamics-single}

It has been quite clearly shown \citep{cai,chubb,golding,raj,sigal,yu} that in a number of experimental
situations some organisms transcribe mRNA discontinuously and as a consequence there is a discontinuous
production of the corresponding effector
 proteins (i.e. protein is produced in bursts).  Experimentally,  the
{\it amplitude} of protein production through bursting translation
of mRNA is exponentially distributed at the single cell level with
density
    \begin{equation}
    h(y) = \dfrac 1 {\bar b} e^{-y/{\bar b}},
    \label{eq:bursting-den}
    \end{equation}
where $\bar b$ is the average burst size, and the {\it frequency}
of bursting $\varphi$ is dependent on the level of the effector.
Writing Equation \ref{eq:bursting-den} in terms of our
dimensionless variables we have
    \begin{equation}
    h(x) = \dfrac 1 {b} e^{-x/{b}}.
    \label{eq:bursting-den-dimen}
    \end{equation}

When bursting is present, the  analog of the
deterministic single slow variable dynamics
discussed  above   is
\begin{equation}
\dfrac{dx}{dt} = -\gamma x + \Xi (h,\varphi(x) ),
\quad \mbox{with} \quad \varphi(x) = \gamma \varphi_{m}
    \mathcal{F}(x),
\label{eq:case1-burst}
\end{equation}
where $\Xi (h,\varphi)$ denotes a jump  Markov process, occurring
at a rate $\varphi$, whose amplitude is distributed with density
$h$ as given in (\ref{eq:bursting-den-dimen}).
Set $\kappa_{b} = \varphi_{m}$, so  $\mathcal{F}$ has the same form as \eqref{eq:composition-x} but with $\kappa_{d,y}$ replaced by $\kappa_{b,y}$.
When we have bursting dynamics described by the stochastic differential equation \eqref{eq:case1-burst}, it has
been shown  \citep{mackeytyran08} that the evolution of the density $u(t,x)$ is governed by the
integro-differential equation
    \begin{equation}
    \begin{split}
    \dfrac{\partial u(t,x)}{\partial t} -\gamma \dfrac{\partial (xu(t,x))}{\partial
    x}& = -\gamma \kappa_{b,x} \mathcal{F}(x)  u(t,x)\\
    & \quad + \gamma \kappa_{b,x} \int_0^x  \mathcal{F}(  z) u(t,z) h(x-z)dz.
    \label{eq:case1-operator-eqn}
    \end{split}
    \end{equation}

In a steady state the (stationary) solution $u_*(x)$ of (\ref{eq:case1-operator-eqn}) is found by solving
\[
    - \dfrac{d (xu_*(x))}{d
    x} = -\kappa_{b,x} \mathcal{F}(  x)u_*(x) + \kappa_{b,x} \int_0^x  \mathcal{F}(  z) u_*(z) h_x(x-z)dz.
    \label{eq:case1-ss-operator-eqn}
    \]
If $u_*(x)$ is unique, then the solution $u(t,x)$ of Equation \ref{eq:case1-operator-eqn} is said to be
asymptotically stable \citep{almcmbk94} in that
 $$
    \lim_{t\to \infty} \int_0^\infty |u(t,x) -u_*(x)|dx = 0
$$
for all initial densities $u(0,x)$.  Somewhat surprisingly, it is possible to actually obtain a closed form solution for $u_*(x)$ as given in the following
\begin{thm}\cite[Theorem 7]{mackeytyran08}. The unique stationary density of Equation
\ref{eq:case1-operator-eqn}, with $\mathcal{F}$ given by Equation
\ref{eq:composition-x} and $h$ given by
(\ref{eq:bursting-den-dimen}), is
    \begin{equation}
    u_*(x) = \dfrac{\mathcal{C}}{x} e^{-x/b} \exp \left [ \kappa_{b,x}
    \int^x \frac{\mathcal{F}(  z)}{z}dz \right
    ],
    \label{eq:ss-soln}
    \end{equation}
$\mathcal{C}$ is a normalization constant such that
$\int_0^\infty u_*(x)dx = 1$,  and  $u(t,x)$ is asymptotically
stable.
\end{thm}
Note that $u_*$ can be written as
\begin{equation}
u_*(x)=\mathcal{C}\exp\int^x\left(\dfrac{\kappa_{b,x} \mathcal{F}(z)}{z}-\frac{1}{b}-\frac{1}{z}\right)dz.
\label{eq:new-ss-soln}
\end{equation}
Thus from \eqref{eq:new-ss-soln}  we can write
\begin{equation}
u_*'(x) = u_*(x) \left ( \dfrac{\kappa_{b,x} \mathcal{F}(x)}{x}-\frac{1}{b}-\frac{1}{x}  \right ),
\label{eq:ss-derivative}
\end{equation}
so for $x > 0$ we  have $u_*'(x) = 0$ if and only if
\begin{equation}
\dfrac {1}{\kappa_{b,x}} \left (  \dfrac {x}{b} + 1 \right ) = \mathcal{F}(x).
\label{eq:cond 1}
\end{equation}
An easy graphical argument shows there may be zero to three positive roots of Equation \ref{eq:cond 1}, and if there are three roots we denote them by $\bar x_1 < \bar x_2 < \bar x_3$.  The  graphical arguments  in conjunction with \eqref{eq:ss-derivative} show that two general cases must be distinguished, exactly as was found in \cite{mty-2011}.   (In what follows, $ \kappa_{b,x}$, $\kappa_{b,x-}$, and $\kappa_{b,x+}$ play exactly the same role as do $ \kappa_{d,x}$, $\kappa_{d,x-}$, and $\kappa_{d,x+}$ in the discussion around Figure \ref{fig:Fig5}.)

\noindent {\bf Case 1. $0<  \kappa_{b,x} < \mathcal{F}_{0}^{-1}$.}  In this case, $u_*(0)=\infty$. If $\kappa_{b,x} < \kappa_{b,x-}$, there are no positive solutions, and $u_* $ will be  a monotone
decreasing function of $x$. If $\kappa_{b,x} > \kappa_{b,x-}$, there are two positive solutions ($\tilde x_2$ and
$\tilde x_3$),  and a maximum in $u_* $ at $\tilde x_3$ with a minimum in $u_* $ at $\tilde x_2$.

\noindent {\bf Case 2. $0 <  \mathcal{F}_{0}^{-1} < \kappa_{b,x} $.}  Now, $u_*(0)=0$ and  either  there are  one, two, or three positive
roots of Equation  \ref{eq:cond 1}. When there are three,  $\tilde x_1, \tilde x_3$ will correspond to the location of maxima in $u_* $ and
$\tilde x_2$ will be the location of the minimum between them.  The  condition for the existence of three roots
is $\kappa_{b,x-} < \kappa_{b,x} < \kappa_{b,x+}$.

Thus we can classify  the stationary density $u_*$ for a bistable switch as:
\begin{enumerate}
\item {\bf Unimodal type 1}: $u_*(0)=\infty$ and $u_*$ is monotone decreasing for
    $0 < \kappa_{b,x} < \kappa_{b,x-}$ and $0 < \kappa_{b,x} <  \mathcal{F}_{0}^{-1}$
\item {\bf Unimodal type 2}: $u_*(0)=0$ and $u_*$ has a single
    maximum at
\begin{enumerate}
    \item $\tilde x_1> 0$ for $  \mathcal{F}_{0}^{-1} < \kappa_{b,x} <
        \kappa_{b,x-} $ or
    \item  $\tilde x_3> 0$ for $\kappa_{b,x+} <
        \kappa_{b,x} $ and $  \mathcal{F}_{0}^{-1} <  \kappa_{b,x}$
    \end{enumerate}
\item {\bf Bimodal type 1}:  $u_*(0)=\infty$ and $u_*$ has
    a single maximum at $\tilde x_3 > 0$ for $\kappa_{b,x-} <
    \kappa_{b,x} <  \mathcal{F}_{0}^{-1}$
\item {\bf Bimodal type 2}:  $u_*(0)=0$ and $u_*$ has two maxima at $\tilde
    x_1,\tilde x_3$,  $0 < \tilde x_1 < \tilde x_3$ for
    $\kappa_{b,x-} < \kappa_{b,x} < \kappa_{b,x+}$ and $  \mathcal{F}_{0}^{-1}  <\kappa_{b,x}$
\end{enumerate}
Note in particular from \eqref{eq:cond 1} that a decrease in the leakage (equivalent to an increase in $\mathcal{F}_{0}^{-1}$) facilitates a transition between unimodal and bimodal stationary distributions and that this is counterbalanced by a increases in the bursting parameters $\kappa_b$ and $b$.  Precisely the same conclusion was obtained by \cite{huang2015} and \cite{ochab2015} on analytic and numerical grounds.

The exact determination of these three roots is difficult in general because of the complexity of  $\mathcal{F}$,  but we can derive implicit criteria for when there are exactly two roots ($\bar x_1$ and $\bar x_3$) by determining when the graph of the left hand side of \eqref{eq:cond 1} is tangent to $\mathcal{F}$.  Using this tangency condition, differentiation of \eqref{eq:cond 1} yields
\begin{equation}
\dfrac {1}{\kappa_{b,x} b} = \mathcal{F}'(x).
\label{eq:cond 2}
\end{equation}

Although Equations \ref{eq:cond 1} and \ref{eq:cond 2} offer conceptually simple conditions for delineating when
there are exactly two roots (and thus to find boundaries between monostable and bistable stationary densities
$u_*$), a moments reflection after looking at \eqref{eq:composition-x} for $\mathcal{F}$ reveals that it
is algebraically quite difficult to obtain quantitative conditions in general.  However, \eqref{eq:cond 1} and
\eqref{eq:cond 2} are easily used in determining numerically boundaries between monostable and bistable
stationary densities.

\subsection{Monomeric repression of one of the genes  with bursting $(n_x = 1)$}
\label{ssec:momomeric}

Evaluation of the integral appearing in Equation \ref{eq:ss-soln} can  be carried out for all (positive) integer
values of $(n_x, n_y)$ in theory, but the calculations become algebraically complicated.  However, if we
consider the situation when a single molecule of the protein from the $y$ gene is capable of repressing the $x$
gene, so $n_x = 1$, then the results become more tractable and allow us to examine the role of different
parameters in determining the nature of $u_*$.

Thus, for $n_x = 1$,   $\mathcal{F}$ takes the  simpler form
\[
\mathcal{F} (x) = \dfrac{
(1+\kappa_{b,y})
+
(\Delta_y + \kappa_{b,y})
 x^{n_y}
}
{
\Lambda + \Gamma x^{n_y}
},
\label{eq:F-monomeric}
\]
where
$$
\Lambda = 1+\Delta_x \kappa_{b,y} > 0, \quad
\Gamma = \Delta_y + \Delta_x \kappa_{b,y} >0.
$$
 Evaluating \eqref{eq:ss-soln}
 we have the explicit representation
\begin{equation}
u_*(x) = {\mathcal{C}} e^{-x/b}
x^{A-1} [\Lambda + \Gamma x^{n_y}]^{\theta}
\label{eq:ss-soln-n=1}
\end{equation}
with
$$
A = \dfrac{\kappa_{b,x}  (1+\kappa_{b,y})}{\Lambda} > 0, \quad
\theta = \dfrac{\kappa_{b,x} \kappa_{b,y} (\Delta _x - 1 ) (\Delta_y - 1) )}{n_y  \Lambda \Gamma} > 0.
$$

\begin{figure}[tb]
\centering \subfigure[]{\includegraphics[width=0.45\textwidth]{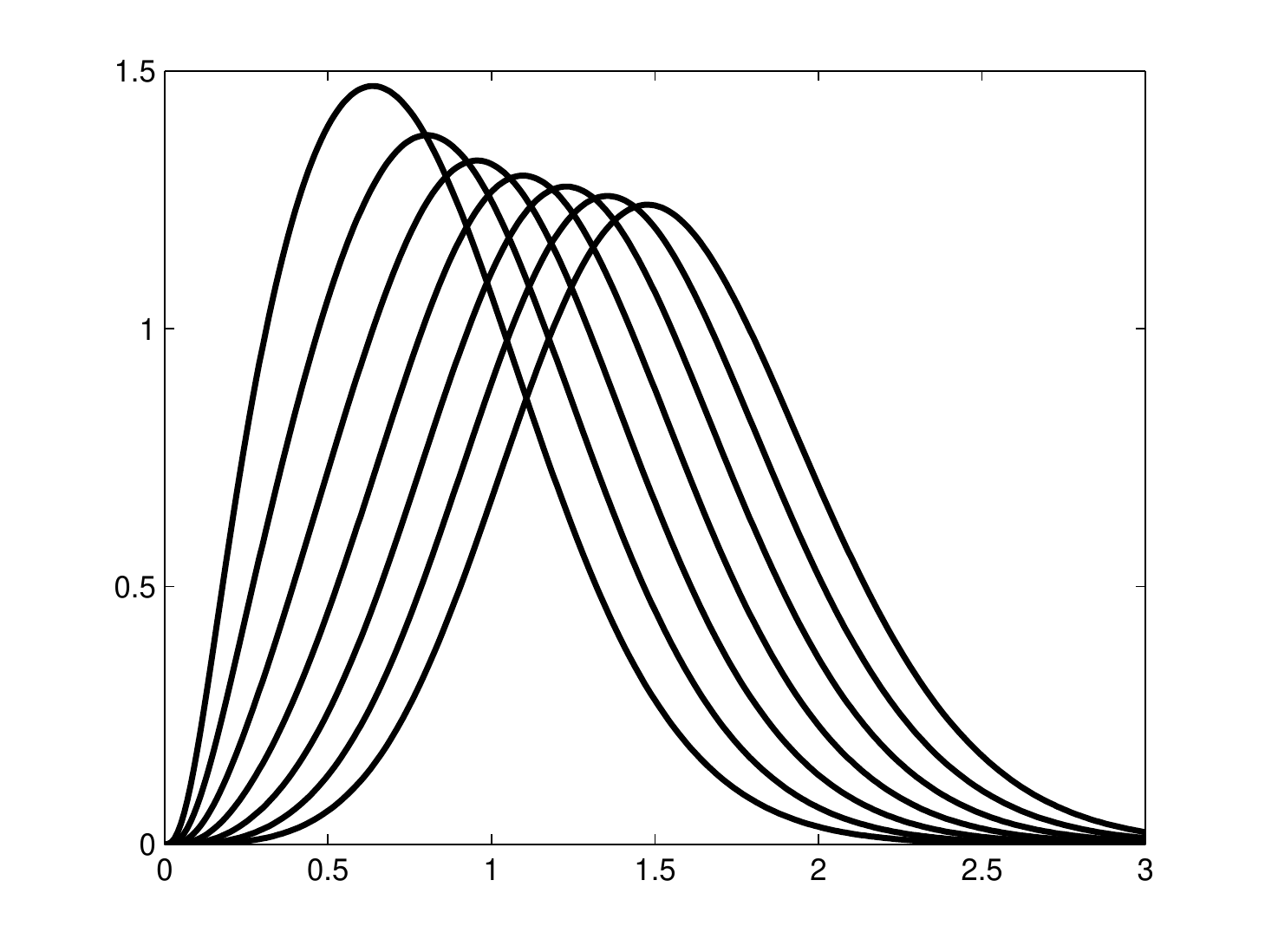}}\qquad
\subfigure[]{\includegraphics[width=0.45\textwidth]{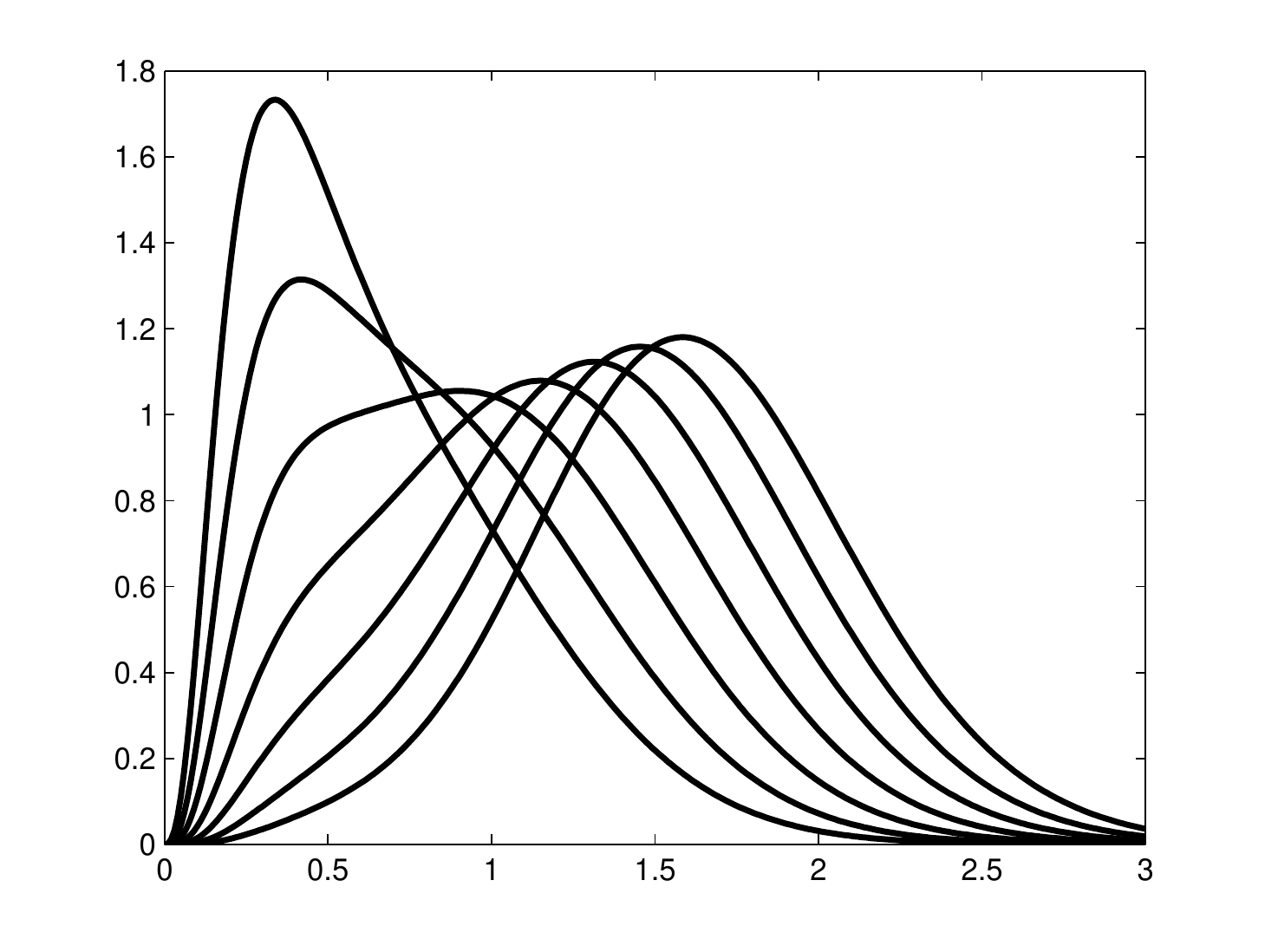}}\\

\subfigure[]{\includegraphics[width=0.45\textwidth]{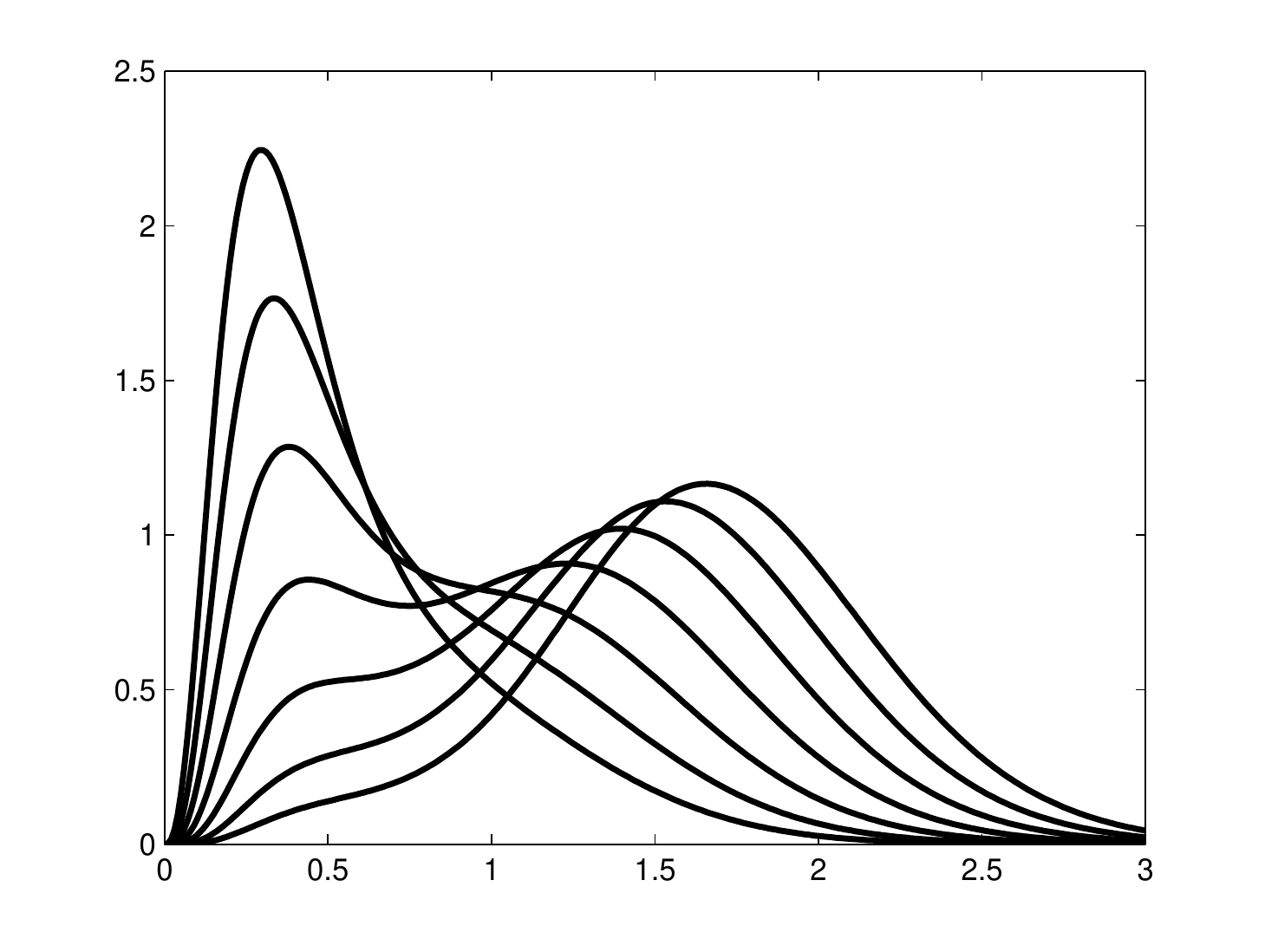}}\qquad
\subfigure[]{\includegraphics[width=0.45\textwidth]{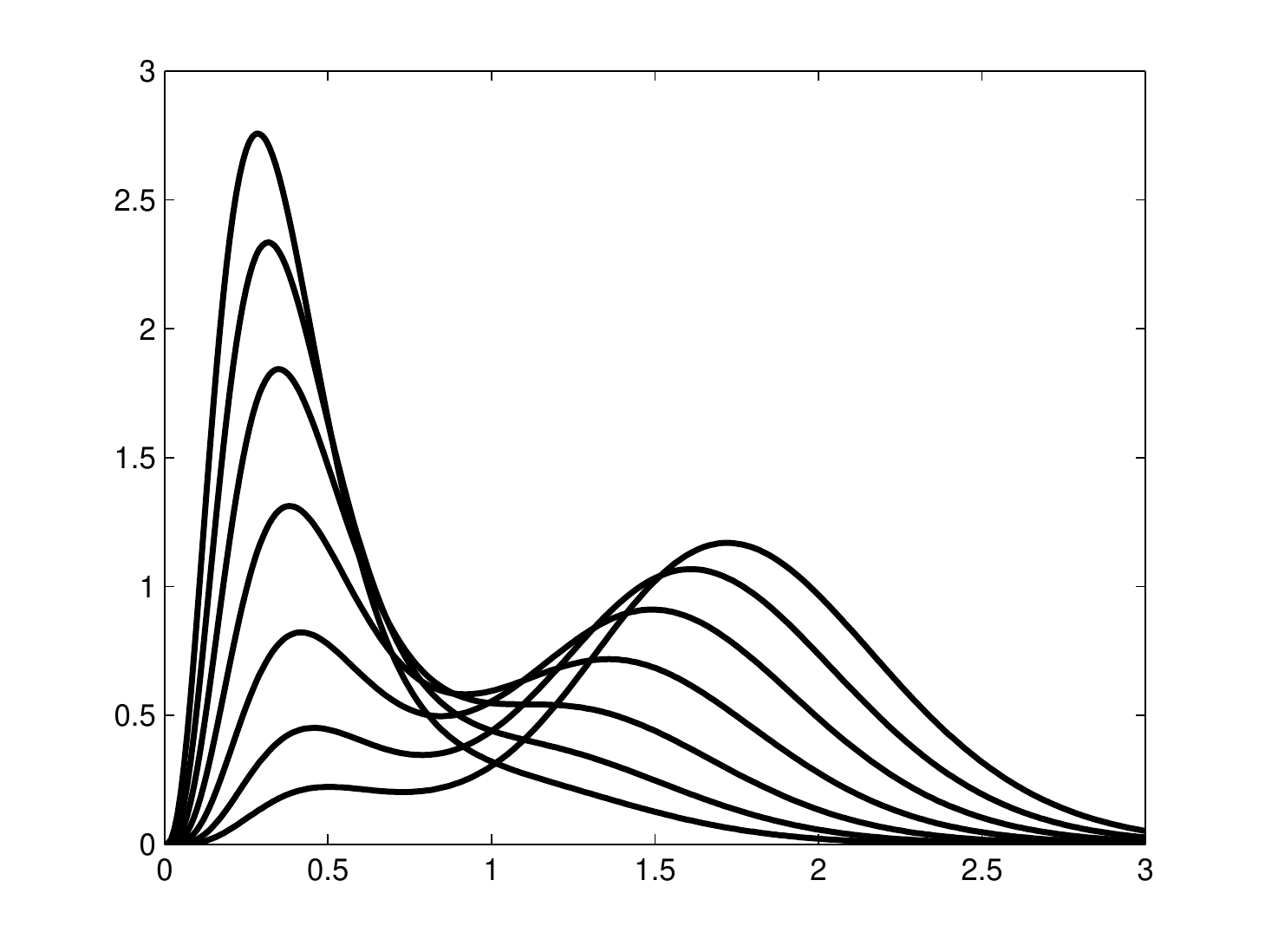}} \caption{In this figure we illustrate stationary
densities given by Equation~\ref{eq:ss-soln-n=1} where the parameter values in each panel are taken to be
$\kappa_{b,y}=1$, $\Delta_x=12$, $\Delta_y=10$, $\kappa_{b,x}\in[25,37]$ changes by 2, where the graph with
highest maximum corresponds to $\kappa_{b,x}=25$ and the maxima are decreasing when  $\kappa_{b,x}$ is
increased. The parameter $n_y$ is taken in an increasing order to be $2,3,4,6$, so that  we start with $n_y=2$
in panel (A) and have $n_y=6$ in panel (D).}
\label{fig:monomeric}
\end{figure}

In Figure \ref{fig:monomeric} we have illustrated the  form of $u_*(x)$ in four different situations.  Figures
\ref{fig:monomeric} A,B show a smooth variation in a Unimodal Type 2 density as $\kappa_{b,x}\in[25,37]$ is
varied by steps of 2 for $n_y = 2$ and $n_y = 3$ respectively.  The behavior is quite different in Figures
\ref{fig:monomeric} C,D however for there, with $n_y = 4$ and $n_y = 6$, the form of  $u_*(x)$ varies from a
Unimodal Type 2 to a Bimodal Type 2 and back again as $\kappa_{b,x}$ is varied.

\subsection{`Bang-bang' repression with bursting}\label{ss:bang-bursting}

We can partially circumvent the algebraic difficulties of the previous sections by considering a
limiting case.  Consider the situation in which $n_y$ becomes large  so $f_y(x)$ approaches the simpler form
\[
f_y(x) \to  \left\{
    \begin{array}{ll}
    1, & 0 \leq x < \theta,  \\
    \Delta_y^{-1}, & \theta \leq x,
     \end{array}
    \right.\label{eq:f-absolute}
\]
where
$$
 \theta \simeq \dfrac{1}{\sqrt[n_y]{\Delta_y}} \sqrt[n_y]{\dfrac{n_y-1}{n_y+1}} \to \dfrac{1}{\sqrt[n_y]{\Delta_y} } \to 1,
$$
so  we have
\begin{equation}
\mathcal{F}(x) \to \left\{
    \begin{array}{ll}
     \mathcal{F}_{0} = \dfrac{1+\kappa_{b,y}^{n_y}}{1 + \Delta_x \kappa_{b,y}^{n_x}} , & 0 \leq x < 1,  \\
     \mathcal{F}_{\infty} = \dfrac{1+(\kappa_{b,y}/\Delta_y)^{n_y}}{1 + \Delta_x (\kappa_{b,y}/\Delta_y)^{n_x}} , & 1  \leq x.
     \end{array}
    \right.\label{eq:absolute}
\end{equation}
The evaluation of \eqref{eq:ss-soln} is simple and yields a stationary density which is (piecewise) that of the gamma distribution:
\[
u_*(x) =  \mathcal{C} e^{-x/b} x^{A(x)-1}
\label{eq:ss-burst-bang}
\]
where
$$
A(x) =  \left\{
    \begin{array}{ll}
     A_0 \equiv {\kappa_{b,x}\mathcal{F}_{0} }  , & 0 \leq x < 1,  \\
     A_\infty \equiv  {\kappa_{b,x}\mathcal{F}_{\infty} }    , & 1  \leq x.
     \end{array}
    \right.\label{eq:ss-burst-bang-A}
$$
Note that $u_*(x)$ is continuous but not differentiable at $x=1$, and  $\mathcal{C}$ is given explicitly by
$$
\mathcal{C} = \dfrac{1}
{b^{A_0 }
[\Gamma ( A_0 )
- \Gamma ( A_0, 1/b )
]
 +
b^{ A_\infty } \Gamma ( A_\infty  )
},
$$
where $\Gamma (\alpha)$ is the gamma function and $\Gamma(\alpha, \beta)$ is the incomplete gamma function.

In this limiting case the stationary density may display one of three general forms as we have classified the densities earlier.  Namely:

\begin{enumerate}
\item If $\kappa_{b,x} < \mathcal{F}_{0}^{-1}$ and $\kappa_{b,x} < \mathcal{F}_{\infty}^{-1}$ then $u_*(x)$ will be of {\bf Unimodal type 1};
\item If $\kappa_{b,x} < \mathcal{F}_{0}^{-1}$ and $\kappa_{b,x} > \mathcal{F}_{\infty}^{-1} $ then $u_*(x)$ will be {\bf Bimodal type 1};
\item If $\kappa_{b,x} > \mathcal{F}_{0}^{-1}$ (which implies  $\kappa_{b,x} > \mathcal{F}_{\infty}^{-1} $) then $u_*(x)$ will be {\bf Bimodal type 2}.
\end{enumerate}

\section{Gaussian distributed noise in the molecular  degradation rate}\label{sec:dynamics-degrad}

For a generic one dimensional stochastic differential equation of
the form
    \begin{equation*}
    dx(t) = \alpha(x) dt + \sigma(x)dw(t),
    \label{eq:generic-sde}
    \end{equation*}
where $w$ is a standard Brownian motion,
the corresponding Fokker Planck equation
    \begin{equation}
    \dfrac
    {\partial u}{\partial t} = - \dfrac {\partial (\alpha u)}{\partial x}
    + \dfrac 12 \dfrac {\partial^2 (\sigma^2u)}{\partial x^2}
    \label{eq:generic-fp}
    \end{equation}
can be written in the form of a conservation equation
    \begin{equation*}
    \dfrac
    {\partial  u}{\partial t} + \dfrac {\partial J}{\partial x} = 0,
    \label{eq:generic-conservation}
    \end{equation*}
where
    \begin{equation*}
    J = \alpha u - \frac 12 \dfrac {\partial (\sigma^2u)}{\partial x}
    \label{eq:current}
    \end{equation*}
is  the probability current.  In a steady state when $\partial_t u
\equiv 0$, the current must satisfy $J = \mbox{constant}$
throughout the domain of the problem.  In the particular case when
$J = 0$ at one of the boundaries (a reflecting boundary) then
$J=0$ for all $x$ in the domain and the steady state solution
$u_*$ of Equation \ref{eq:generic-fp} is easily obtained with a
single quadrature as
    \begin{equation*}
    u_*(x) = \dfrac {\mathcal C}{\sigma^2(x)} \exp
    \left \{
    2 \int^x \dfrac {\alpha (y)}{\sigma^2(y)}dy
    \right \},
    \label{eq:generic-ss}
    \end{equation*}
where $\mathcal C$ is a normalizing constant as before.

In our considerations of the  effects of continuous fluctuations, we examine the  situation in which Gaussian
fluctuations appear in the degradation rate $\gamma_x$ of the generic equation (\ref{eq:one-slow}).
\cite{gillespie2000} has shown that in this situation we need to consider what he calls the chemical Langevin equation, so \eqref{eq:one-slow} takes the form
    \begin{equation*}
    dx = \gamma [\kappa_{d,x} \mathcal{F} (  x)  - x] dt +     \sqrt{\gamma x}
    dw.
    \label{eq:1d-stoch-decay}
    \end{equation*}
(In the situation we consider here, $\alpha (x) =  \gamma [ \kappa_{d,x}  \mathcal{F} (  x) - x]$ and $\sigma(x)
=  \sigma_{\gamma} \sqrt{x}$.) Within the Ito interpretation of stochastic integration,  this equation has a
corresponding Fokker Planck equation for the evolution of the ensemble density $u(t,x)$ given by
    \begin{equation}
     \dfrac
     {\partial u}{\partial t} = - \dfrac {\partial
    \left [\gamma
    (\kappa_{d,x}  \mathcal{F} (  x) -  x) u
    \right ]}{\partial x}
    + \dfrac {{\gamma} }{2} \dfrac {\partial ^2 ( x u )}{\partial x^2}.
    \label{eq:1d-stoch-decay-fp}
    \end{equation}
Since  concentrations of molecules cannot become negative the
boundary at $x=0$ is reflecting and the stationary solution of
Equation \ref{eq:1d-stoch-decay-fp} is given by
    \begin{equation}
    u_*(x) = \dfrac
{\mathcal{C}}{x}
e^{-2 x  }
    \exp \left [   2 \kappa_{d,x} \int^x \frac{ \mathcal{F} (z)}{z}dz \right
    ]. 
    \label{eq:ss-soln-decay}
    \end{equation}

We have also the following result.

\begin{thm}\cite[Theorem
2]{pichorrudnicki00}. The unique stationary density of Equation
\ref{eq:1d-stoch-decay-fp} is given by Equation
\ref{eq:ss-soln-decay}. Further $u(t,x)$ is asymptotically
stable.
\end{thm}

\begin{rem}
Note that the stationary solution for the density $u_*(x)$ given by Equation \ref{eq:ss-soln-decay} in the presence of noise in the protein degradation rate is {\bf identical} to the solution in Equation \ref{eq:ss-soln}, when transcriptional and/or translational noise in present in the system, as long as we make the identification of $\kappa_{b,x}$ with $ 2\kappa_{d,x}  $ and $b$ with $ 1/2 $.  As a consequence, all of the results of the analysis in Section \ref{sec:dynamics-single} are applicable in this section.  { The implication is, of course, that one cannot distinguish between the location of the noise simply based on the nature of the stationary density.}\end{rem}

\section{Two dominant slow genes with bursting}
\label{sec:2-dom-burst}

In this last section we turn our attention to the situation in which we have two slow variables, one in each gene. If there are two slow variables with one in each of the $X$ and $Y$ genes, then we obtain a two dimensional system that is significantly different and more difficult to deal with from what we have encountered so far,
and we wish to examine the existence of the stationary density $u_*(x,y)$ in the presence of bursting production.

For two dominant slow variables in different genes with bursting, the stochastic analogs of the deterministic equations are
\begin{align*}
    \dfrac{dx}{dt} &= -\gamma_{x}x + \Xi (h_1,\varphi_1(y) ) \quad \mbox{with} \quad \varphi_1(y) = \gamma_x \kappa_{b,x} f_x( y), \label{eq:case3-burst-x}\\
    \dfrac{dy}{dt} &= -\gamma_{y}y + \Xi (h_2,\varphi_2(x) ) \quad \mbox{with} \quad \varphi_2(x) = \gamma_y \kappa_{b,y} f_y( x). 
    \end{align*}
To be more specific, let $x(t)$ and $y(t)$ denote the amount of protein in a cell at time $t$, $t\ge 0$, produced by gene $X$ and $Y$, respectively.
If only degradation were present, then $(x(t),y(t))$ would satisfy the equation
\begin{equation}\label{e:flow}
x'(t)=-\gamma_x x(t), \quad y'(t)=-\gamma_y y(t), \quad t\ge 0.
\end{equation}
The solution of \eqref{e:flow} starting at time $t_0=0$ from  $(x_0,y_0)\in \mathbb{R}_+^2$ is of the form
\[
\pi_t(x_0,y_0)=(e^{-\gamma_x t}x_0,e^{-\gamma_y t}y_0), \quad t\ge 0.
\]
But, we interrupt the degradation at random
times
\[
0=t_0<t_1<t_2<\ldots
\]
when, independently of everything else,
a random amount
of protein $x$ or $y$ is produced
according to an exponential distribution with mean $b_x$ or $b_y$, respectively, with densities
\begin{equation*}
    h_1(x) = \dfrac 1 {b_x} e^{-x/{b_x}},\quad  h_2(y) = \dfrac 1 {b_y} e^{-y/{b_y}}.
    \label{e:burst}
    \end{equation*}
The rate of production of protein $x$ (protein $y$) depends on the level of protein $y$ (protein $x$) and is $\varphi_1(y)$ ($\varphi_2(x)$). Consequently,  at each $t_k$ if $x(t_k)=x$ and $y(t_k)=y$ then
 one of the genes $X$ or $Y$ can be chosen at random  with probabilities $p_1$ or $p_2$, respectively, given by
\[
p_1(x,y)=\frac{\varphi_1(y)}{\varphi(x,y)},\quad p_2(x,y)=\frac{\varphi_2(x)}{\varphi(x,y)},
\]
and we have
\[
\Pr(t_{k+1}-t_{k}>t|x(t_{k})=x,y(t_{k})=y)=e^{-\int_0^t \varphi(\pi_s(x,y))ds},\quad t>0,
\]
where the function $\varphi$ is of the form
\[
\varphi(x,y)=\varphi_1(y)+\varphi_2(x)=\gamma_x \kappa_{b,x} f_x(y)+ \gamma_y \kappa_{b,y} f_y (x).
\]

The process $Z(t)=(x(t),y(t))$ is a Markov process with values in $E=[0,\infty)^2=\mathbb{R}_+^2$ given by
\[
Z(t)=\left\{
       \begin{array}{ll}
       \pi_{t-t_{k-1}}(Z(t_{k-1})), &  t_{k-1}\le t< t_k, \\
        Z(t_{k}-)+\xi_k, & t=t_k, k=1,2,\ldots
       \end{array}
     \right.
\]
where
\[
Z(t_{k}-)=\pi_{t_k-t_{k-1}}(Z(t_{k-1}))
\]
and
$(\xi_k)_{k\ge 1}$ is a sequence of random variables such that
\begin{equation*}
\Pr(Z(t_{k}-)+\xi_k\in B|Z(t_{k}-)=z)= \mathcal{P}(z,B)
\end{equation*}
with
\[
\mathcal{P}(z,B)=p_1(z)\int_0^{\infty}1_B(z+\theta e_1)h_1(\theta)d\theta\\ +p_2(z)\int_0^{\infty}1_B(z+\theta e_2)h_2(\theta)d\theta.
\]
Here $e_1$ and $e_2$ are the unit vectors from $\mathbb{R}^2$
\[
e_1=\left(
      \begin{array}{c}
        1 \\
        0 \\
      \end{array}
    \right),\quad
    e_2=\left(
          \begin{array}{c}
            0 \\
            1 \\
          \end{array}
        \right).
\]
Let $\mathbb{P}_z$ be the distribution  of the process $Z=\{Z(t)\}_{t\ge 0}$ starting at  $Z(0)=z$ and  $\mathbb{E}_z$ the corresponding expectation operator.
For any $z$ and any Borel subset of $\mathbb{R}^2_+$ we have
\[
\mathbb{P}_z(Z(t)\in B)=\sum_{n=0}^\infty \mathbb{P}_z(Z(t)\in B, t_n\le t< t_{n+1}).
\]
If the distribution of $Z(0)$ has a probability density $u_0$ with respect to the Lebesgue measure $m$ on $\mathbb{R}_+^2$ then $Z(t)$ has the distribution with density $P(t)u_0$, i.e.,
\begin{equation}\label{e:density}
\int_E  \mathbb{P}_z(Z(t)\in B)u_0(z)m(dz)=\int_B P(t)u_0(z)m(dz), \quad B\in  \mathcal{B}(\mathbb{R}_+^2).
\end{equation}
The evolution equation for the density $u(t,x,y)=P(t)u_0(x,y)$  is
\begin{equation}
\begin{split}
\label{eq:case3-operator-eqn}
 \dfrac{\partial u}{\partial t} -\gamma_x\dfrac{\partial (x u)}{\partial x} - \gamma_y\dfrac{\partial (y u)}{\partial y}& = -\varphi(x,y) u(t,x,y) \\
&\quad+  \varphi_1(y) \int_0^{x}  h_1(x-z_x)u(t,z_x,y)dz_x  \\
&\quad+\varphi_2(x) \int_0^{y}   h_2(y-z_y)u(t,x,z_y)dz_y,
\end{split}
\end{equation}
with initial condition $u(0,x,y)=u_0(x,y)$, $x,y\in [0,\infty)$.

\begin{thm} There is a unique  density $u_*(x,y)$ which is a stationary solution of \eqref{eq:case3-operator-eqn} and  $u(t,x,y)$ is asymptotically stable.
\end{thm}

\begin{proof}
We use the notation of \cite{rudnickipichortyran02} and apply  \cite[Theorem 5]{rudnickipichortyran02} together with \cite[Theorem 1]{pichorrudnicki00}. Let $E=\mathbb{R}^2_+$ and $m$ be the Lebesgue measure on $E$. The evolution equation \eqref{eq:case3-operator-eqn} induces a strongly continuous semigroup $\{P(t)\}_{t\ge 0}$ of Markov operators on the space of Lebesgue integrable functions $L^1=L^1(E,\mathcal{B}(E),m)$ (see e.g. \cite{tyran09}).
Recall that a function $f\colon E\to \mathbb{R}_+$ is called lower semicontinuous, if
\[
\liminf_{w\to z}f(w)\ge f(z)
\]
for every $z\in E$.
We show  that there exists a nonnegative Borel function $q$ defined on $(0,\infty)\times E\times E$  with the following properties
\begin{enumerate}
\item for each $t>0$ and each Borel set $B$
\[
\mathbb{P}_z(Z(t)\in B)\ge \int_B q(t,w,z)m(dw)\quad\text{for all } z\in E,
\]
\item for each $t>0$ the function $(w,z)\mapsto q(t,w,z)$ is lower semicontinuous,
\item for each $z$ there exists $t>0$ such that
\[
\int_{E}q(t,w,z)m(dw)>0,
\]
\item  for $m$-a.e. $z\in E$ and every Borel set $B$ with $m(B)>0$
\[
\int_0^{\infty} \int_B q(t,w,z)m(dw)dt>0.
\]
\end{enumerate}
Then it follows from \cite[Theorem 5]{rudnickipichortyran02} and  \cite[Theorem 1]{pichorrudnicki00}  that either $u(t,x,y)$ is asymptotically stable or the process $Z$ is sweeping from compact subsets of $E$, i.e.,
\begin{equation}\label{e:sweep}
\lim_{t\to\infty}\int_E \mathbb{P}_z(Z(t)\in F)u_0(z)m(dz)=0
\end{equation}
for all compact sets $F\subset E$ and all densities $u_0$.
We have
\[
\begin{split}
\mathbb{P}_z(Z(t)\in B)=\sum_{n=0}^\infty \mathbb{P}_z(\pi_{t-t_n}Z(t_n)\in B, t_n\le t< t_{n+1}).
\end{split}
\]
The discrete time process $(Z(t_n),t_n)_{n\ge 0}$ is Markov with transition probability $P((z,s),B\times I)=Q_z(B\times ((I-s)\cap \mathbb{R}_+))$ for Borel subsets $B$ of $\mathbb{R}^2_+$ and Borel subsets $I$ of $\mathbb{R}_+$,
where $Q$  is given by
\begin{multline*}
Q_z(B\times I)=\mathbb{P}_z(Z(t_1)\in B, t_1\in I)= \\
\sum_{i=1}^2\int_{I}\int_{0}^\infty 1_B(\pi_{s_1}z+\theta e_{i})\varphi(\pi_{s_1}z)e^{-\phi_z(s_1)}p_{i}(\pi_{s_1}z)h_{i}(\theta)d\theta ds_1.
\end{multline*}
 We have
 \begin{equation*}
\mathbb{P}_z(\pi_{t-t_1}(Z(t_1))\in B, t_1\le t< t_{2})=
\int_{E\times [0,t]} 1_B(\pi_{t- s_1}z_1) e^{-\phi_{z_1}(t-s_1)}Q_z(dz_1,ds_1)
\end{equation*}
and for $k=2$ we obtain
\begin{multline*}
\mathbb{P}_z(\pi_{t-t_2}(Z(t_2))\in B, t_2\le t< t_{3})=\\
\int_{E\times [0,t]} \int_{E\times [0,t-s_1]}1_B(\pi_{t-(s_2+s_1)}z_2) e^{-\phi_{z_2}(t-(s_2+s_1))}Q_{z_1}(dz_2,ds_2)Q_z(dz_1,ds_1).
\end{multline*}
Since $\varphi$ is bounded from above by a constant $\overline{\varphi}$ and $\varphi p_i=\varphi_i$ is bounded from below by a constant $c_i>0$, we obtain that
\begin{equation*}
Q_z(B\times I)\ge
\int_{I}\int_{0}^\infty\sum_{i=1}^2 1_B(\pi_{s_1}z+\theta e_{i})e^{-\overline{\varphi} s_1}c_i h_{i}(\theta)d\theta ds_1
\end{equation*}
for all $z$.
Now, if $z_1=\pi_{s_1}z+\theta_1e_{1}$ and $z_2=\pi_{s_2}z_1+\theta_2e_{2}$, then
\[
\begin{split}
\pi_{t-(s_2+s_1)}z_2&=\pi_tz+\pi_{t-s_1}(\theta_1e_{1})+\pi_{t-(s_2+s_1)}(\theta_2e_{2})\\
&=\pi_tz+T_{(s_1,s_2)}(\theta_1,\theta_2),
\end{split}
\]
where
\[
T_{(s_1,s_2)}(\theta_1,\theta_2)=(\theta_1e^{-\gamma_1(t-s_1)}, \theta_2e^{-\gamma_2(t-(s_2+s_1))}).
\]
Consequently, we obtain
\begin{multline*}
\mathbb{P}_z(Z(t)\in B)\ge\\ \int_{0}^t \int_{0}^{t-s_1}\int_{0}^{\infty}\int_{0}^{\infty} 1_B(\pi_{t}z+T_{(s_1,s_2)}(\theta_1,\theta_2))e^{-\overline{\varphi}t}c_1c_2h_1(\theta_1)h_2(\theta_2)d\theta_2d\theta_1 ds_2ds_1.
\end{multline*}
The transformation $(\theta_1,\theta_2)\mapsto T_{(s_1,s_2)}(\theta_1,\theta_2)$ is invertible on $(0,\infty)^2$, thus we can make a change of variables under the integral to conclude that
\[
\mathbb{P}_z(Z(t)\in B)\ge \int_{(0,\infty)^2} 1_{B}(\pi_tz+w)e^{-(\overline{\varphi}+\gamma_1+\gamma_2)t}\widetilde{q}(t,w)dw,
\]
where
\begin{multline*}
\widetilde{q}(t,(w_1,w_2))=\\ \int_{0}^t \int_{0}^{t-s_1}e^{\gamma_1 s_1+\gamma_2(s_1+s_2)}c_1c_2
h_1(e^{\gamma_1(t-s_1)}w_1)h_2(e^{\gamma_2(t-s_1-s_2)}w_2)ds_2ds_1.
\end{multline*}
Consequently, we obtain
\[
\mathbb{P}_z(Z(t)\in B)\ge \int_{B} q(t,w,z) dw,
\]
where
\[
q(t,w,z)=e^{-(\overline{\varphi}+\gamma_1+\gamma_2)t}\widetilde{q}(t,w-\pi_tz)1_{(0,\infty)^2}(w-\pi_tz), \quad w,z\in E.
\]
For each $t>0$ the function $(w,z)\mapsto q(t,w,z)$ is lower semicontinuous and
\[
\int_{E}q(t,w,z)dw>0
\]
for every   $z$. Finally, to check the last condition note that $\pi_t z$ converges to zero as $t\to \infty$ for every $z$. Thus, for every $z\in E$ and $w\in (0,\infty)^2$ we can find  $t_0>0$ such that $w-\pi_tz\in  (0,\infty)^2$ for every $t\ge t_0$, which implies that
\[
\int_0^\infty q(t,w,z)dt>0
\]
for all $z\in E$ and $w\in (0,\infty)^2$.

Next, we show that the process is not sweeping from compact subsets. Suppose, contrary to our claim, that the process is sweeping. It follows from \eqref{e:sweep} that for every compact set $F$ and every density $u_0$ we have
\[
\lim_{t\to\infty}\frac{1}{t}\int_0^t \int_E \mathbb{P}_z(Z(s)\in F)u_0(z)m(dz)=0.
\]
Chebyshev inequality implies that
\[
\mathbb{P}_z(Z(t)\in F_a)\ge 1-\frac{1}{a}\mathbb{E}_zV(Z(t))
\]
for all $t>0$, $z\in E$, and $a>0$, where $V$ is a nonnegative measurable function and $F_a=\{z\in E: V(z)\le a\}$. To get a contradiction it is enough to show that
 \[
\limsup_{t\to\infty}\frac{1}{t}\int_0^t\int_E \mathbb{E}_zV(Z(s))u_0(z)m(dz)ds<\infty
\]
for a density $u_0$ and a continuous function $V$ such that each  $F_a$ is a compact subset of $E$.
Recall that an operator  $\mathcal{L}$ is the extended generator of the Markov process $Z$, if
its domain  $\mathcal{D}(\mathcal{L})$ consists of those measurable $V\colon E\to \mathbb{R}$ for which there exists a measurable $U\colon E\to \mathbb{R}$ such that for each $z\in E$, $t>0$,
\[
\mathbb{E}_z(V(Z(t)))=V(z)+\mathbb{E}_z\left(\int_{0}^t U(Z(s))\,ds\right)
\]
and
\[
\int_{0}^t \mathbb{E}_z(|U(Z(s))|)ds<\infty,
\]
in which case we define $\mathcal{L}V=U$. From \cite[Theorem 26.14 and Remark 26.16]{davis93} it follows that
\[
\mathcal{L}V(z)=\mathcal{L}_0V(z)+\varphi(z)\int_E( V(w)-V(z))\mathcal{P}(z,dw),
\]
where for $z=(x,y)$ we have
\[
\mathcal{L}_0V(x,y)=-\gamma_1 x\frac{\partial V}{\partial x}(x,y)-\gamma_2 y\frac{\partial V}{\partial y}(x,y)
\]
and that
$V$  belongs to the domain of $\mathcal{L}$ if the function $t\mapsto V(\pi_t(x,y))$  is absolutely continuous and for each $t$ 
\[
\mathbb{E}(\sum_{n: t_n\le t}|V(Z(t_n))-V(Z(t_n-))|)<\infty.
\]
Observe that for  $V(x,y)=x+y$ this condition holds and we obtain
\[
\mathcal{L}V(x,y)=-(\gamma_x+\gamma_y)V(x,y)+ \varphi(x,y) (b_x p_1(x,y)+b_y p_2(x,y)).
\]
Consequently, there are positive constants $c_1$ and $c_2$ such that for all $z=(x,y)$ we have
\[
\mathcal{L}V(z)\le -c_1 V(z)+ c_2,
\]
which implies that
\[
\mathbb{E}_z(V(Z(t)))\le V(z)-c_1\int_{0}^t\mathbb{E}_zV(Z(s))\,ds+c_2 t.
\]
Hence, for each $t>0$ and $z\in E$ we have
\[
\frac{1}{t}\int_{0}^t\mathbb{E}_zV(Z(s))\,ds\le \frac{c_2}{c_1}+\frac{1}{c_1t}V(z).
\]
Taking a density $u_0$ with $\int_E V(z)u_0(z)m(dz)<\infty$  completes the proof.
\end{proof}

\begin{rem}
Observe that
\[
\frac{\partial }{\partial x}\int_0^{x}  h_1(x-z_x)u_*(z_x,y)dz_x=\frac{1}{b_x}\left(u_*(x,y)-\int_0^{x}  h_1(x-z_x)u_*(z_x,y)dz_x\right).
\]
Thus the equation for the stationary density $u_*(x,y)$ can be rewritten as
\begin{multline*}
\frac{\partial }{\partial x}\left(\gamma_x xu_*(x,y)-\varphi_1(y) e^{-x/b_x}\int_0^x e^{z_x/b_x} u_*(z_x,y)dz_x\right)\\
+\frac{\partial }{\partial y}\left(\gamma_y yu_*(x,y)-\varphi_2(x) e^{-y/b_y}\int_0^y e^{z_y/b_y} u_*(x,z_y)dz_y\right)=0.
\end{multline*}
However, we have been unable to find an analytic solution to this equation.
\end{rem}

\section{Discussion and conclusions}

Here we have considered the behavior of a bistable molecular switch in both  its deterministic version as well
as what happens in the presence of two different kinds of noise.  The results that we have obtained in the
presence of noise are, unfortunately, only partial due to the analytic difficulties in solving for the
stationary density but we have been able to offer analytic expressions for  $u_*(x)$ either in the presence of transcriptional and/or translational bursting (Section \ref{sec:dynamics-single}) or in the presence of Gaussian noise on the degradation rate (Section \ref{sec:dynamics-degrad}) when there is a single dominant slow variable. { We have shown that in both cases one cannot distinguish between the source of the noise based on the nature of the stationary density.} In the situation where there are two dominant slow variables (Section \ref{sec:2-dom-burst}) we have established the asymptotic stability of $u(t,x,y)$, and thus the uniqueness of the stationary density $u_*(x,y)$.

\begin{acknowledgements}
This work was supported by the Natural Sciences and Engineering Research Council (NSERC, Canada)  and the Polish NCN grant no 2014/13/B/ST1/00224.  We are grateful to Marc Roussel (Lethbridge), Romain Yvinec (Tours) and
Changjing Zhuge (Tsinghua University, Beijing) for helpful comments on this problem.  We are especially indebted to the referees and the Associate Editor for their comments that have materially improved this paper.
\end{acknowledgements}

%
%
%

\end{document}